\documentclass[aps,pra,groupedaddress,twocolumn,notitlepage,bibnotes]{revtex4-1}

\usepackage{amsthm,amsmath,amssymb,amsbsy}

\usepackage{hyperref}
\usepackage{url}

\usepackage{graphicx}
\usepackage{bm,bbm}%
\usepackage{braket}
\usepackage{enumerate}
\usepackage{booktabs,multirow}
\usepackage{mathtools,bbding}
\usepackage{microtype}
\usepackage{mathrsfs}

\usepackage{txfonts,newtxmath}

\usepackage[most,breakable]{tcolorbox}

\newcommand{\nc}{\newcommand}
\newcommand{\rnc}{\renewcommand}

\makeatletter
\newcommand*\rel@kern[1]{\kern#1\dimexpr\macc@kerna}
\newcommand*\widebar[1]{%
  \begingroup
  \def\mathaccent##1##2{%
    \rel@kern{0.8}%
    \overline{\rel@kern{-0.8}\macc@nucleus\rel@kern{0.2}}%
    \rel@kern{-0.2}%
  }%
  \macc@depth\@ne
  \let\math@bgroup\@empty \let\math@egroup\macc@set@skewchar
  \mathsurround\z@ \frozen@everymath{\mathgroup\macc@group\relax}%
  \macc@set@skewchar\relax
  \let\mathaccentV\macc@nested@a
  \macc@nested@a\relax111{#1}%
  \endgroup
}

\rnc{\thesection}{\arabic{section}}
\rnc{\thesubsection}{\thesection.\arabic{subsection}}
\rnc{\thesubsubsection}{\thesubsection.\arabic{subsubsection}}

\usepackage{etoolbox}
\hypersetup{linktoc=all}

\newtheorem{definition}{Definition}
\newtheorem{proposition}[definition]{Proposition}
\newtheorem{lemma}[definition]{Lemma}

\DeclareMathOperator{\Tr}{Tr}

\nc{\tcr}[1]{{\color{red} #1}}
\nc{\tcb}[1]{{\color{blue} #1}}

\nc{\psic}{\psi^{c}}
\nc{\two}[1]{\underline{2^{d-#1}}}
\rnc{\H}{\mathcal{H}}
\nc{\Hanc}{\mathcal{H}_{\text{anc}}}
\nc{\psianc}{\psi_{\text{anc}}}
\nc{\lampow}{\lambda^{1/d}}

\nc{\proj}[1]{\ket{#1}\!\bra{#1}}
\nc{\pro}[1]{#1 #1^\dagger}
\nc{\RR}{{{\mathbb R}}}
\nc{\CC}{{{\mathbb C}}}
\nc{\FF}{{{\mathbb F}}}
\nc{\NN}{{{\mathbb N}}}
\nc{\ZZ}{{{\mathbb Z}}}

\nc{\MIO}{{\text{\rm MIO}}}
\nc{\DIO}{{\text{\rm DIO}}}
\nc{\SIO}{{\text{\rm SIO}}}
\nc{\IO}{{\text{\rm IO}}}

\nc{\SEP}{{\text{SEP}}}
\nc{\NS}{{\text{NS}}}
\nc{\LOCC}{{\text{LOCC}}}
\nc{\PPT}{{\text{PPT}}}
\nc{\EXT}{{\text{EXT}}}
\nc{\OLOCC}{{\text{1-LOCC}}}
\nc{\SEPP}{{\text{SEPP}}}

\nc{\MC}{{\text{\rm MC}}}

\nc{\cE}{\mathscr{E}}

\rnc{\*}{\textup{*}}
\nc{\ketbra}[1]{\ket{#1}\!\!\bra{#1}}

\newcommand{\ketbraa}[2]{\ket{#1}\!\!\bra{#2}}

\newcommand{\id}{\mathbbm{1}}

\nc{\MM}{\widetilde{\M}}
\nc{\Ml}{\M^{\leq}}

\nc{\mleq}{\preceq}
\nc{\mgeq}{\succeq}

\nc{\ox}{\otimes}

\nc{\wt}{\widetilde}

\nc{\SDP}{\text{\rm SDP}}

\nc{\cc}{{\circ\circ}}
\nc{\mnorm}[1]{\norm{#1}{[m]}}

\nc{\F}{\mathcal{F}}
\nc{\M}{\mathcal{M}}

\let\oldproofname\proofname
\rnc{\proofname}{\rm\bf{\oldproofname}}
\rnc{\qedsymbol}{{\color{gray!50!black}\rule{0.6em}{0.6em}}}

\newcommand{\bb}{\begin{equation}}
\newcommand{\bbb}{\begin{equation*}}
\newcommand{\ee}{\end{equation}}
\newcommand{\eee}{\end{equation*}}

\nc{\note}[1]{{\color{blue!90!black} #1}}

\newcommand{\notts}{\affiliation{School of Mathematical Sciences and Centre for the Mathematics and Theoretical Physics of Quantum Non-Equilibrium Systems, University of Nottingham, University Park, Nottingham NG7 2RD, United Kingdom}}

\begin{document}

\title{Assisted Work Distillation}

\author{Benjamin Morris}
\email{benjamin.morris@nottingham.ac.uk}
\notts

\author{Ludovico Lami}
\email{ludovico.lami@gmail.com}
\notts

\author{Gerardo Adesso}
\email{gerardo.adesso@nottingham.ac.uk}
\notts

\begin{abstract}
We study the process of assisted work distillation. This scenario arises when two parties share a bipartite quantum state $\rho_{AB}$ and their task is to locally distil the optimal amount of work when one party is restricted to thermal operations whereas the other can perform general quantum operations and they are allowed to communicate classically. We demonstrate that this question is intimately related to the distillation of classical/quantum correlations. In particular, we show that the advantage of one party performing global measurements over many copies of $\rho_{AB}$ is related to the non-additivity of the entanglement of formation. We also show that there may exist work bound in the quantum correlations of the state that is only extractable under the wider class of local Gibbs-preserving operations.
\end{abstract}

%\date{\today}
\maketitle

%%%%%%%%%%%%%%%%%%%%%%%%%%%%%%%%%%%%%%%%%%%%%%%%%%%%%%%%%%%%%%%%%%%%%%%%%%%%%%%%%%%%%%%%%%%%%%%%%%%%%%%%%%%%%%%
%%%%%%%%%%%%%%%%%%%%%%%%%%%%%%%%%%%%%%%%%%%%%%%%%%%%%%%%%%%%%%%%%%%%%%%%%%%%%%%%%%%%%%%%%%%%%%%%%%%%%%%%%%%%%%%
%%%%%%%%%%%%%%%%%%%%%%%%%%%%%%%%%%%%%%%%%%%%%%%%%%%%%%%%%%%%%%%%%%%%%%%%%%%%%%%%%%%%%%%%%%%%%%%%%%%%%%%%%%%%%%%

\textbf{\em Introduction}.--- The recently conceived field of quantum thermodynamics represents a drive to understand the interplay of the two fundamental theories of thermodynamics and quantum mechanics. Scientists from various disciplines such as open quantum systems~\cite{alicki2018introduction}, stochastic thermodynamics~\cite{elouard2015stochastic} and information theory~\cite{goold2016role} are utilizing their respective tools to answer these fundamental questions. In particular, recent work~\cite{janzing2000thermodynamic, brandao2013resource} has demonstrated that thermodynamics may be understood from a resource-theoretic perspective, allowing researchers to investigate thermodynamic transformations in a quantum information setting.

\begin{figure}[t]
	\centering
	\includegraphics[width=8cm]{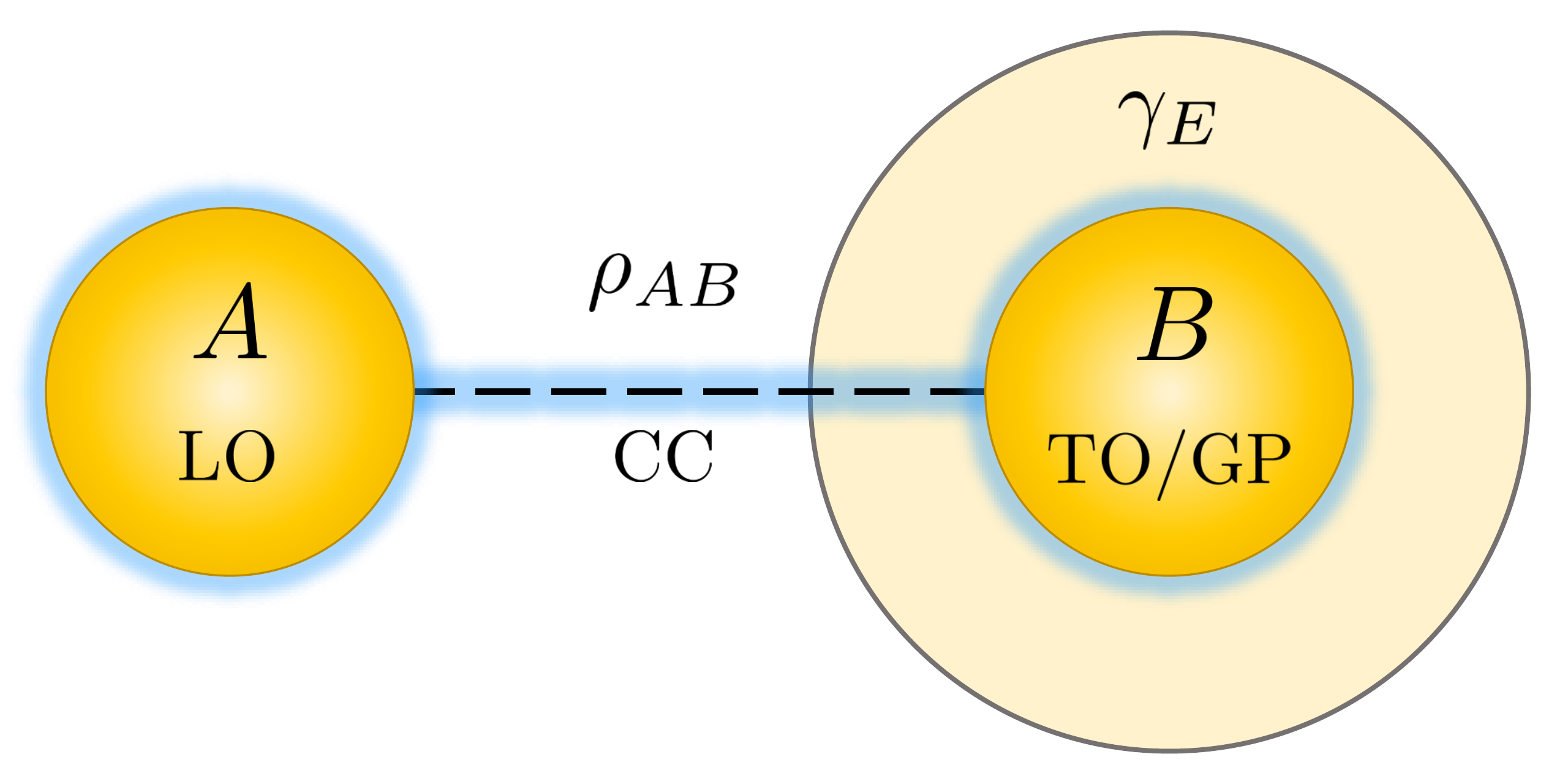}
	\caption{We investigate distillation of work from a quantum system $B$ controlled by an observer, Bob, who is constrained to thermal operations or Gibbs-preserving operations, and is assisted by another party, Alice, who can perform arbitrary local operations on an ancillary system $A$ and communicate classically with Bob. The \textit{work of assistance} and the \textit{work of collaboration} are defined and related to the correlations in the state $\rho_{AB}$ shared by Alice and Bob.}
	\label{figdiag}
\end{figure}

In this paper, we investigate the task of {\em assisted work distillation}, see Fig.~\ref{figdiag}. Here, the process of work distillation is intended in a resource theoretic framework to be the asymptotic distillation of reference states with energy but no entropy by means of thermal operations, meaning that the  {\em distillable} (or {\em extractable}) {\em work} can be quantified by how distinguishable a quantum state is from a Gibbs equilibrium state~\cite{brandao2013resource} --- for other definitions of work in quantum thermodynamics see e.g.~\cite{goold2016role}. In the assisted scenario, two parties, Alice ($A$) and Bob ($B$), share many copies of a bipartite state $\rho_{AB}$. Between them their goal is to maximize the quantity of distillable work on Bob's subsystem.
Alice may perform arbitrary quantum operations on her subsystem whereas Bob is restricted to thermal operations on his. By utilizing correlations within $\rho_{AB}$ and classical communication between the parties we demonstrate key features of Bob's distillable work.

In particular, we characterize the set of shared states which allow for local work distillation. We also demonstrate that for a protocol involving one-way communication between the parties explicit expressions for the local distillable work, which we dub the \textit{work of assistance} (in analogy with the \textit{entanglement of assistance}~\cite{divincenzo1999entanglement}), can be derived both in the regularized and un-regularized scenarios. From these expressions we make use of two central results from quantum information theory to show that Alice performing global measurements over many copies of the shared state offers an explicit advantage over single copy measurements. We also show that this advantage disappears  when the initial state is pure.

In addition to the \textit{work of assistance} we also define the \textit{work of collaboration}, defined as such to allow two-way communication between the parties and local Gibbs-preserving operations~\cite{faist2015gibbs} on Bob's side. We show that by allowing this collaboration and the wider class of operations, the local distillable work can increase. We also demonstrate that for an initial pure state the \textit{work of collaboration} may yield an increase in distillable work by an amount proportional to the entropy of Bob's subsystem $S(\rho_{B})$, where $\rho_{B}=\Tr_{A}\left[\rho_{AB}\right]$.

It is important to consider the realm in which our results apply. Within the resource theoretic framework it is typical to consider resource inconvertibility in the asymptotic scenario. This is particularly pertinent for thermodynamics due to its equivalence to taking the thermodynamic limit, which suppresses the appearance of fluctuations.
%In this context the distillable work of a state $\rho$  may be understood in terms of how far the state $\rho$ is from a thermal equilibrium state $\gamma$~\cite{brandao2013resource}.

\textbf{\em Resource theories of thermodynamics}. --- %
We start by explicitly defining what is meant by the resource theory of thermal operations (TO). Originally introduced by~\cite{janzing2000thermodynamic,brandao2013resource} the allowed operations for a quantum system $S$ with Hilbert space $\mathcal{H}$ and Hamiltonian $H_S$ are the completely-positive trace-preserving (CPTP) maps $\mathcal{E}:\mathcal{L}(\mathcal{H})\rightarrow\mathcal{L}(\mathcal{H})$ of the form
\begin{align}\label{TO}
	\mathcal{E}(\rho)=\Tr_{E}\left(U_{SE}\left(\rho_S\otimes\gamma_E\right)U_{SE}^\dagger\right),
\end{align}
where $U_{SE}$ is an arbitrary unitary operation, acting jointly on the system $S$ and a reservoir $E$, that commutes with the global Hamiltonian $\left[U,H_S\otimes\id_E+\id_S\otimes H_E\right]=0$, and  $\gamma={Z}^{-1}e^{-\beta H}$ denotes the Gibbs thermal equilibrium state at inverse temperature $\beta$ and partition function ${Z}$.
The joint unitary operations and partial trace define the \textit{free operations} of the resource theory whereas the Gibbs states define the \textit{free states}. By explicitly accounting for the resources used, the TO framework provides a general setting within which to study thermodynamic transformations, in particular the distillation of work.

In this setting, following~\cite{horodecki2013fundamental} we define the distillable work from a system $B$ in the state $\rho_B$ as the maximum number $RE$ such that the transformations $\rho_B^{\otimes n} \otimes \ketbra{0}_P^{\otimes [Rn]}\to \ketbra{1}_P^{\otimes [Rn]}$ are possible with TO at background inverse temperature $\beta$ with asymptotically vanishing error. Here, referring to (\ref{TO}), we are considering a composite system $S$ which consists of the principal system $B$ with Hamiltonian $H_B$ and a qubit battery $P$ with Hamiltonian $H_P\coloneqq E\ketbra{1}_P$, where $E$ is a free parameter we are allowed to optimize over. In formula,
\begin{eqnarray} \label{distillablework}
W(\rho_B)\! &\coloneqq &\sup\bigg\{ RE: \\[-.2cm]
&&\!\lim_{n\to\infty}\! \inf_{\Lambda\in\mathrm{TO}}\left\| \Lambda\Big(\rho_B^{\otimes n}\!\!\otimes\! \ketbra{0}_P^{\otimes [Rn]} \Big)\! - \ketbra{1}_P^{\otimes [Rn]} \right\|_1\!=\! 0 \bigg\}. \nonumber
\end{eqnarray}
It follows from the main result of~\cite{brandao2013resource} (see appendix~A in Supplemental Material~\cite{EPAPS} for an explicit derivation) that the distillable work defined in Eq.~\eqref{distillablework} equals the change in \textit{free energy}:
\begin{equation} \label{workisfreeenergy}
    W(\rho_B) \equiv \Delta F(\rho_B) = \frac{1}{\beta} S(\rho_B\|\gamma_B) ,
\end{equation}
with $S(\rho\|\gamma)=\Tr\left(\rho \log \rho - \rho \log \gamma\right)$ being the relative entropy of athermality. Observe that $S(\rho\|\gamma)$ is monotonically non-increasing under TO.

A larger class of operations are Gibbs-Preserving (GP) operations; these are CPTP maps $\Lambda$ that admit as their fixed point the Gibbs state at a given temperature, i.e.\ such that $\Lambda(\gamma_B)=\gamma_B$. The motivation behind this alternative framework that regards GP operations as free operations for thermodynamics, is that any non-GP operation, $\Lambda(\gamma)=\sigma\neq\gamma$, could be used to extract an arbitrarily large amount of work from $\sigma^{\otimes n}$ as $n\rightarrow\infty$. It can be clearly seen from~\eqref{TO} that TO are a subset of GP, and the inclusion is known to be strict~\cite{faist2015gibbs}.

\textbf{\em Work of assistance}. --- In this section we consider the case where Alice and Bob have access to the shared state $\rho_{AB}$ and we allow one-way classical communication from Alice to Bob. This is similarly motivated as the recently studied `conditioned thermal operations'~\cite{narasimhachar2017resource}. Alice, whom operations are unrestricted, may perform on her subsystem the positive operator-valued measurement (POVM) $\{\Pi_{A,i}\}$, whose associated probabilities are $p_i=\Tr\left[\rho_A \Pi_{A,i}\right]$, whereas Bob is restricted to TO. Alice performing her measurement and communicating the outcome to Bob results in him having access to the ensemble $\{p_i, \tilde{\rho}_{B,i}\}$, where
\begin{align} \label{rho tilde}
	\tilde{\rho}_{B,i}=\frac{1}{p_i}\Tr_{A}\left[\left(\Pi_{A,i}\otimes\id_B\right) \rho_{AB}\right].
\end{align}
In the scenario we consider, Alice's goal is to help Bob to distil as much work as possible.
From this train of thoughts we define our first quantity of interest, the \textit{work of assistance},
\begin{align}\label{workass1}
	W_{a}^{B|A}(\rho_{AB}):=\max_{\{\Pi_{A,i}\}}\frac{1}{\beta}\sum_i p_i S(\tilde{\rho}_{B, i}||\gamma_B),
\end{align}
where the maximization is taken over the set of Alice's measurements (i.e.\ POVMs). Using convexity, we see that this quantity is lower bounded by $\frac{1}{\beta}S(\rho_B||\gamma_B)$, which of course means that  being assisted by Alice is generally no worse than having no assistance at all. Not only: as we show in appendix~B~\cite{EPAPS}, all states $\rho_{AB}$ that exhibit some form of correlation, i.e.\ such that $\rho_{AB}\neq \rho_A\otimes \rho_B$ is not factorized, satisfy the strict inequality $W_a^{B|A}(\rho_{AB})> \frac{1}{\beta} S(\rho_B\|\gamma_B)$, implying that there is an assisted protocol that helps Bob distilling more work. In particular, the states from which Bob can distil no work at all even in the assisted setting are simply products of the form $\Gamma_{AB}=\sigma_A\otimes\gamma_B$, from now on referred to as \textit{quantum-thermal} (QT) states, the same states have been found in the conditional thermal operations setting~\cite{narasimhachar2017resource}.

In appendix~C~\cite{EPAPS} we show that $W_{a}^{B|A}$ can be written as
\begin{align}\label{workass4}
W_a^{B|A}(\rho_{AB})=\frac{1}{\beta}\left(S(\rho_B||\gamma_B) + J^{\rightarrow}(\rho_{AB})\right),
\end{align}
where $J^{\rightarrow}(\rho_{AB})$ is the Henderson--Vedral~\cite{henderson2001classical} measure of classical correlations (with respect to measurements on Alice) defined as
$J^{\rightarrow}(\rho_{AB})\coloneqq \max_{\{\Pi_{A,i}\}}\left(S(\rho_B)-  \sum_ip_iS(\tilde{\rho}_{B,i})\right)$. The result in equation~\eqref{workass4} clearly separates the quantity of work distillable by Bob with/without the assistance of Alice. This is in agreement with a recent result in~\cite{PhysRevLett.121.120602}.

An important question to ask is whether this quantity of work changes if Alice is able to perform measurements over many copies of the shared initial state $\rho_{AB}$. In order to answer this question we continue by defining the regularized work of assistance,
\begin{align}\label{regularized}
W_{a, \infty}^{B|A}(\rho_{AB})\coloneqq \lim_{n\rightarrow \infty}\frac{1}{n}W_a^{B|A} \left(\rho_{AB}^{\otimes n}\right).
\end{align}
In appendix~D~\cite{EPAPS} we show that the above quantity indeed yields the best achievable rate of work distillation in the case where the only allowed communication is from Alice to Bob.
Although the regularization makes it hard to compute, the r.h.s.\ of~\eqref{regularized} can nonetheless be related to a quantifier known as \textit{distillable common randomness} $C_D$, introduced in~\cite{henderson2001classical} as
\begin{align}\label{CD}
C_D^{\rightarrow}(\rho_{AB})=\lim_{n\rightarrow \infty}\frac{1}{n}J^{\rightarrow}\left(\rho_{AB}^{\otimes n}\right),
\end{align}
and then interpreted operationally in~\cite{devetak2004distilling}.  The operational interpretation of $C_D$ rests on protocols that extract from $n$ independent copies of $\rho_{AB}$ a total of $C$ maximally correlated classical bits  via $R$ bits of noiseless classical communication between Alice and Bob with vanishing error. The quantity $C_D$ is thus defined as the maximum net gain $(C-R)/n$ in the limit $n\rightarrow\infty$. For a discussion from the thermodynamical point of view, see~\cite{oppenheim2002thermodynamical}.

Using the definition in equation~\eqref{CD} and the fact that the relative entropy is additive, we can therefore write the regularized work of assistance as,
\begin{align}\label{workassinf}
W_{a, \infty}^{B|A}(\rho_{AB})=\frac{1}{\beta}\left(S(\rho_B||\gamma_B)+C_D^{\rightarrow}(\rho_{AB})\right),
\end{align}
again clearly separating the quantity of distillable work with/without the assistance of Alice.

Upon defining the regularized version of $W_a^{B|A}(\rho_{AB})$ we should ask whether giving Alice the ability to perform global measurements over many copies of the shared state $\rho_{AB}$ increases the average work that Bob can distil. In order to answer this question we employ two fundamental results from the field of quantum information. On the one hand,~\cite[Theorem~1]{koashi2004monogamy} states that
\begin{align}
\label{Eform}E_f(\rho_{A'B})+J^{\rightarrow}(\rho_{AB})&=S(\rho_B),\\
\label{Ecost}E_C(\rho_{A'B})+C^{\rightarrow}_D(\rho_{AB})&=S(\rho_B),
\end{align}
provided that $\rho_{A'B}$ is the $A$-complement of $\rho_{AB}$, i.e.\ there exists a pure state extension $\rho_{AA'B}$ that satisfies $\Tr_A\left[\rho_{AA'B}\right]=\rho_{A'B}$ and $\Tr_{A'}\left[\rho_{AA'B}\right]=\rho_{AB}$. Here, $E_f(\rho_{AB})$ stands for the entanglement of formation~\cite{bennett1996purification}, while the entanglement cost is given by $E_C(\rho_{AB})=\lim_{n\rightarrow \infty}\frac{1}{n}E_f(\rho_{AB}^{\otimes n})$, and quantifies the amount of Bell states needed to form $\rho_{AB}$ via LOCC protocols in the asymptotic limit of many copies~\cite{hayden2001asymptotic}.

Substituting equations~\eqref{Eform}--\eqref{Ecost} into equations~\eqref{workass4}--\eqref{workassinf} respectively allows us to write $W_a^{B|A}(\rho_{AB})$ and $W_{a, \infty}^{B|A}(\rho_{AB})$ in terms of these entanglement measures,
\begin{align}
W_{a}^{B|A}(\rho_{AB})=&\frac{1}{\beta}\big(S(\rho_B||\gamma_B)+S(\rho_B)-E_{f}(\rho_{A'B})\big)\, , \label{wef}\\
W_{a, \infty}^{B|A}(\rho_{AB})=&\frac{1}{\beta}\big(S(\rho_B||\gamma_B)+ S(\rho_B)-E_C(\rho_{A'B})\big)\, . \label{wec}
\end{align}
This allows us to take advantage of another fundamental result of quantum information, the non-additivity of $E_f(\rho_{AB})$~\cite{hastings2008counterexample}. Therefore, despite the additivity of the (relative) von Neumann entropy we can state that the ability for Alice to perform global measurements can increase the amount of work Bob can distil, i.e.\ for some states $\rho_{AB}$ it will happen that
\begin{align} \label{neq}
	W_{a}^{B|A}(\rho_{AB}) < W_{a, \infty}^{B|A}(\rho_{AB}).
\end{align}
However, for many simple classes of states the above does not happen. For instance, in appendix~E~\cite{EPAPS} we explicitly calculate $W_{a}^{B|A}$ for the  relevant family of isotropic states in arbitrary dimension, and show its additivity over multiple copies.
\\\textbf{\em Work of collaboration }. ---
Let us consider an arbitrary class of operations $\mathcal{O}$ on a thermodynamical system. We assume that $\mathcal{O}$ contains not only deterministic operations, but also so-called \textit{quantum instruments}, i.e.\ collections $\{\Phi_i\}_i$ of completely positive maps such that $\sum_i \Phi_i$ is trace-preserving. Physically, the classical label $i$ will record the outcomes of the quantum measurements that have been made throughout the process, while $\Tr\Phi_i(\rho)$ represents the probability of the outcome $i$ occurring when the state $\rho$ is processed.
In a bipartite setting, we can construct the associated set $\mathcal{O}_c^{B|A}$ of collaborative operations by concatenating in any order: (1) instruments in $\mathcal{O}$ on $B$; (2) classical communication between Alice and Bob; (3) arbitrary quantum operations on $A$. We can now define the associated \textit{work of collaboration}
in analogy with Eq.~\eqref{distillablework} as
\begin{eqnarray} \label{collab1}
&&W_{c}^{B|A}\!\big(\rho_{AB}\big) \!\coloneqq \sup\bigg\{ RE: \\[-0.2cm]
&&\qquad \lim_{n\rightarrow \infty}\! \inf_{\Lambda\in\mathcal{O}_c^{B|A}}\left\|\Lambda\Big(\rho_{AB}^{\otimes n}\otimes \ketbra{0}_P^{\otimes [Rn]}\Big)-\ketbra{1}_P^{\otimes[Rn]}\right\|_1 \!=0 \bigg\} , \nonumber
\end{eqnarray}
where it is understood that the battery $P$ pertains to Bob's system, and its Hamiltonian is again given by  $H_P\coloneqq E\ketbra{1}_P$, with $E$ a free parameter.

%where $\psi\coloneqq \ketbra{\psi}$. %The reference state is chosen to be pure, i.e.\ with vanishing entropy, so that energy and free energy coincide and represent the available work.
%Since the resource theory of TO is reversible, the choice of the reference state is entirely arbitrary. We can declare by convention that $\psi_B$ has unit distillable work, i.e.\ $\frac{1}{\beta}S\left( \psi_B\|\gamma_B\right)=1$.

By their very definition~\eqref{TO}, TO are intrinsically deterministic. Therefore, in the collaborative setting there is no information Bob can send to Alice if he is restricted to TO, and the corresponding work of collaboration reduces to the regularized work of assistance as given in~\eqref{regularized}.
To investigate the collaborative setting in greater detail it is thus indispensable to expand Bob's allowed operations to the wider class~\cite{faist2015gibbs} of GP operations, that satisfy $\Lambda(\gamma_B)=\gamma_B$. This less restrictive framework crucially allows Bob to apply non-deterministic instruments $\{\Phi_i\}_i$, which are required to satisfy $\Phi_i(\gamma_B)\propto \gamma_B$ for all $i$. The outcome $i$ can then be communicated to Alice via the classical communication channel.

From now on, we will therefore consider the work of collaboration~\eqref{collab1} as defined for the collaborative set of operations $\mathrm{GP}_c^{B|A}$ corresponding to GP operations on Bob. It is clear that QT states of the form $\Gamma_{AB}=\sigma_{A}\otimes\gamma_B$, where $\sigma_A$ is arbitrary, can be generated for free even in the TO framework. Furthermore, it can be shown that these are all the states for which $W_c^{B|A}(\rho_{AB})=0$. This suggests the following definition of the \textit{relative entropy of collaboration},
\begin{align} \label{relative1}
	W_{r}^{B|A}(\rho_{AB})\coloneqq \frac{1}{\beta}\min_{\sigma_{A}}S\left(\rho_{AB}\|\sigma_A\otimes \gamma_B\right),
\end{align}
where the minimization is taken over the set QT. In appendix~F~\cite{EPAPS} we explicitly demonstrate monotonicity of this function under the set of allowed operations. We also prove in appendix~F~\cite{EPAPS} that the minimization in~\eqref{relative1} can be explicitly solved so as to give
\begin{align} \label{relative2}
W_{r}^{B|A}(\rho_{AB})=\frac{1}{\beta}S\left(\rho_{AB}\|\rho_A\otimes\gamma_B\right).
\end{align}
Simple algebraic manipulations allow us to recast this as
\begin{align}\label{workcollabup}
	W_{r}^{B|A}(\rho_{AB})=\frac{1}{\beta}\left(S(\rho_{B}||\gamma_B)+I(\rho_{AB})\right),
\end{align}
where $I(\rho_{AB})\coloneqq S(\rho_A)+S(\rho_B) - S(\rho_{AB})$ is the mutual information quantifying total correlations between Alice and Bob.
\\\textbf{\em Comparing measures of assistance}. ---
Equation~\eqref{workcollabup} suggests that the mutual information quantifies the amount by which the collaboration between the parties increases Bob's distillable work.

In fact, we are able to demonstrate in appendix~G~\cite{EPAPS} that $W_{r}^{B|A}$ provides an upper bound on the work of collaboration. We can also observe that since TO are a subset of GP operations, the work of collaboration is no smaller than the regularized work of assistance. This can also be deduced by comparing~\eqref{workassinf} with~\eqref{workcollabup}, and using the well-known fact that $C^{\rightarrow}_D(\rho_{AB})\leq I(\rho_{AB})$~\cite{devetak2001low, devetak2004distilling}.
Putting all together:
\begin{align}\label{altogether}
	W_a^{B|A}(\rho_{AB})\le W_{a,\infty}^{B|A}(\rho_{AB}) \le W_c^{B|A}(\rho_{AB})\le W_{r}^{B|A}(\rho_{AB}).
\end{align}

Recall from~\eqref{neq} that there can be a strict inequality between the  two leftmost quantities in the above chain of inequalities. Concerning the two rightmost ones, quite interestingly, we find that the gap $W_r^{B|A}(\rho_{AB})- W_{a}^{B|A}(\rho_{AB})$ is explicitly described by the \textit{quantum discord}, a measure of the quantumness of the correlations between Alice and Bob~\cite{ollivier2001quantum, henderson2001classical}.
Indeed, by comparing equations~\eqref{workass4} and~\eqref{workcollabup}, we find
\begin{align}\label{discordia}
W_r^{B|A}(\rho_{AB})- W_{a}^{B|A}(\rho_{AB}) &= \frac{1}{\beta}\left( I(\rho_{AB})-
J^{\rightarrow}(\rho_{AB})\right) \nonumber\\&=: \frac{1}{\beta}D^{\rightarrow}(\rho_{AB}),
\end{align}
where $D^{\rightarrow}(\rho_{AB})$ is the quantum discord, quantifying the share of correlations lost between Alice and Bob as a consequence of a minimally disturbing measurement on Alice's side. This result shows that the \textit{work of collaboration} can exceed the \textit{work of assistance} by an amount bounded from above by the shared quantum correlations, measured by the discord $D^{\rightarrow}(\rho_{AB})$.   We note that recent works~\cite{francica2017daemonic, PhysRevLett.121.120602} has suggested a protocol for explicitly distilling the work locked in the quantum discord, however the operations considered lie outside those in TO.
Other interpretations for the  quantum discord in thermodynamical and related contexts have also been explored  in the literature~\cite{zurek2003quantum, oppenheim2002thermodynamical,liuzzo2016thermodynamics,adesso2016measures}.

It is particularly instructive to analyze all the quantities appearing in equation~\eqref{altogether} for the relevant case where Alice holds a purification of Bob's state, i.e.\ $\rho_{AB}=\phi_{AB}=\ketbra{\phi}_{AB}$. On the one hand, for a pure state $\phi_{AB}$ it is known~\cite{henderson2001classical, devetak2004distilling} that the Henderson--Vedral measure and distillable common randomness coincide with the local entropy of each subsystem, i.e.\ $J^{\rightarrow}(\phi_{AB}) = C_{D}^{\rightarrow}(\phi_{AB}) = S(\phi_{B})$. Hence,
\begin{align}\label{workasspure}
	W_{a}^{B|A}(\phi_{AB})= W_{a, \infty}^{B|A}(\phi_{AB})=\frac{1}{\beta}\left(S(\rho_{B}||\gamma_B)+S(\rho_B)\right),
\end{align}
implying that for an initial pure state the ability for Alice to perform global measurements over many copies gives no advantage in Bob distilling work. On the other hand, it is also elementary to verify that
\begin{align}\label{workcollabpure}
	W_{r}^{B|A}(\phi_{AB})=\frac{1}{\beta}\left(S(\rho_{B}||\gamma_B)+2\,S(\rho_B)\right).	
\end{align}
Therefore by comparing equations~\eqref{workasspure} and~\eqref{workcollabpure} it is seen that for an initial pure state we demonstrate that relaxing the local operations from TO to GP map might allow Bob to distil a bound quantity of work equal to the local entropy.

\textbf{\em Conclusion}. --- In this work we have fully characterized the task of assisted work distillation in the asymptotic scenario of quantum thermodynamics, addressing  questions left open in ~\cite{chitambar2016assisted, oppenheim2002thermodynamical}. In particular we have introduced two relevant quantities of interest, the \textit{work of assistance} and the \textit{work of collaboration}. These quantities allowed us to investigate the possible advantage of local GP operations over TO and global measurements on a system; in particular, how GP operations may allow Bob to locally distil the work bound within the quantum correlations of the initial shared state.

Although it was shown that GP operations can provide an increase in distillable work, the explicit relationship between the \textit{work of assistance} and the \textit{work of collaboration} requires further investigation, as for the latter quantity only an upper bound was derived here. We further stress that our results only hold in the asymptotic limit. It would be interesting to investigate assisted work distillation in the {\em single-shot regime}, to determine the role correlations  play in work fluctuations. This could prove useful for near-term technological applications.

The present analysis adds to the literature on assisted distillation of different quantum resources~\cite{divincenzo1999entanglement,horodecki2003local,devetak2005distillation,chitambar2016assisted,bartosz-myself-alex,streltsov2017towards}. In particular, Refs.~\cite{chitambar2016assisted, bartosz-myself-alex} studied the distillation of quantum coherence~\cite{streltsov2017colloquium}, rather than work, from Bob's system with the assistance of Alice. In that setting, Bob is limited to incoherent operations~\cite{baumgratz2014quantifying} while Alice can perform arbitrary local quantum operations, and the two parties can communicate classically. We can draw a comparison between the two settings, by noting that the additional quantity of resource that can be distilled from Bob's system thanks to Alice's assistance amounts to the entropy of Bob's reduced state in the case of coherence~\cite{chitambar2016assisted} and to the classical correlations shared between Alice and Bob in the case of work [Eq.~(\ref{workass4})].
We can further observe how the hierarchy presented in (\ref{altogether}) for assisted work distillation is analogous to the one derived in~\cite{chitambar2016assisted} for assisted coherence distillation, but the key role of quantum discord in bounding the gap between work of assistance and work of collaboration is only revealed in this paper by comparing the power of different classes of local operations for Alice (TO versus GP). It would be meaningful to revisit the assisted coherence distillation framework by imposing additional physical constraints on Alice's operations, e.g.~by adopting strictly incoherent operations~\cite{winter2016operational} or TO, and hence exploiting the methods developed in this paper for the characterization of other quantum resources.

Our findings could have implications for the understanding of the Szilard engine~\cite{szilard1929entropieverminderung}. The latter is a simple physical model which demonstrates how information may be exploited in order to extract physical work. The relevance of this model was then understood in the context of information processing by Landauer~\cite{landauer1961irreversibility}. Many recent works have discussed the application of a Szilard engine in quantum thermodynamics~\cite{kim2011quantum, mohammady2017quantum, cottet2017observing, park2013heat, zurek1986maxwell, reeb2014improved, sagawa2008second}, deriving bounds for work extraction that are related to~\eqref{workcollabup}~\cite{park2013heat, zurek1986maxwell, reeb2014improved, sagawa2008second} in a setting where a second party, historically entitled Maxwell's Demon, is in possession of a state correlated to the thermodynamic system. The converse setting, where correlations can be formed from initially uncorrelated states using thermal operations has also been studied~\cite{narasimhachar2018quantifying}.

The results presented here provide further links between the fields of quantum information and thermodynamics. In particular, how highly studied  measures of information provide us with an insight into the thermodynamics of correlations. These results both contribute to our knowledge of the fundamental nature of thermodynamics but also may become essential for the thermodynamic control of a quantum computer.

\begin{acknowledgments}
\textbf{\em Acknowledgments.}--- We are grateful to Harry Miller, Philippe Faist, Luis Correa, Carlo Maria Scandolo, Mark Whitworth, Ryuji Takagi, Francesco Plastina, and Paul Hollywood for helpful discussions. We acknowledge financial support from the European Research Council (ERC) under the Starting Grant GQCOP (Grant No.~637352) and the EPSRC (Grant No.~EP/N50970X/1).
\end{acknowledgments}

%\bibliographystyle{apsrev4-1}
%\bibliography{Index}

\begin{thebibliography}{50}%
\makeatletter
\providecommand \@ifxundefined [1]{%
 \@ifx{#1\undefined}
}%
\providecommand \@ifnum [1]{%
 \ifnum #1\expandafter \@firstoftwo
 \else \expandafter \@secondoftwo
 \fi
}%
\providecommand \@ifx [1]{%
 \ifx #1\expandafter \@firstoftwo
 \else \expandafter \@secondoftwo
 \fi
}%
\providecommand \natexlab [1]{#1}%
\providecommand \enquote  [1]{``#1''}%
\providecommand \bibnamefont  [1]{#1}%
\providecommand \bibfnamefont [1]{#1}%
\providecommand \citenamefont [1]{#1}%
\providecommand \href@noop [0]{\@secondoftwo}%
\providecommand \href [0]{\begingroup \@sanitize@url \@href}%
\providecommand \@href[1]{\@@startlink{#1}\@@href}%
\providecommand \@@href[1]{\endgroup#1\@@endlink}%
\providecommand \@sanitize@url [0]{\catcode `\\12\catcode `\$12\catcode
  `\&12\catcode `\#12\catcode `\^12\catcode `\_12\catcode `\%12\relax}%
\providecommand \@@startlink[1]{}%
\providecommand \@@endlink[0]{}%
\providecommand \url  [0]{\begingroup\@sanitize@url \@url }%
\providecommand \@url [1]{\endgroup\@href {#1}{\urlprefix }}%
\providecommand \urlprefix  [0]{URL }%
\providecommand \Eprint [0]{\href }%
\providecommand \doibase [0]{http://dx.doi.org/}%
\providecommand \selectlanguage [0]{\@gobble}%
\providecommand \bibinfo  [0]{\@secondoftwo}%
\providecommand \bibfield  [0]{\@secondoftwo}%
\providecommand \translation [1]{[#1]}%
\providecommand \BibitemOpen [0]{}%
\providecommand \bibitemStop [0]{}%
\providecommand \bibitemNoStop [0]{.\EOS\space}%
\providecommand \EOS [0]{\spacefactor3000\relax}%
\providecommand \BibitemShut  [1]{\csname bibitem#1\endcsname}%
\let\auto@bib@innerbib\@empty
%</preamble>
\bibitem [{\citenamefont {Alicki}\ and\ \citenamefont
  {Kosloff}(2018)}]{alicki2018introduction}%
  \BibitemOpen
  \bibfield  {author} {\bibinfo {author} {\bibfnamefont {R.}~\bibnamefont
  {Alicki}}\ and\ \bibinfo {author} {\bibfnamefont {R.}~\bibnamefont
  {Kosloff}},\ }\href@noop {} {\bibfield  {journal} {\bibinfo  {journal} {arXiv
  preprint arXiv:1801.08314}\ } (\bibinfo {year} {2018})}\BibitemShut {NoStop}%
\bibitem [{\citenamefont {Elouard}\ \emph {et~al.}(2017)\citenamefont
  {Elouard}, \citenamefont {Herrera-Mart{\'\i}}, \citenamefont {Clusel},\ and\
  \citenamefont {Auff{\`e}ves}}]{elouard2015stochastic}%
  \BibitemOpen
  \bibfield  {author} {\bibinfo {author} {\bibfnamefont {C.}~\bibnamefont
  {Elouard}}, \bibinfo {author} {\bibfnamefont {D.~A.}\ \bibnamefont
  {Herrera-Mart{\'\i}}}, \bibinfo {author} {\bibfnamefont {M.}~\bibnamefont
  {Clusel}}, \ and\ \bibinfo {author} {\bibfnamefont {A.}~\bibnamefont
  {Auff{\`e}ves}},\ }\href@noop {} {\bibfield  {journal} {\bibinfo  {journal}
  {npj Quantum Information}\ }\textbf {\bibinfo {volume} {3}},\ \bibinfo
  {pages} {9} (\bibinfo {year} {2017})}\BibitemShut {NoStop}%
\bibitem [{\citenamefont {Goold}\ \emph {et~al.}(2016)\citenamefont {Goold},
  \citenamefont {Huber}, \citenamefont {Riera}, \citenamefont {del Rio},\ and\
  \citenamefont {Skrzypczyk}}]{goold2016role}%
  \BibitemOpen
  \bibfield  {author} {\bibinfo {author} {\bibfnamefont {J.}~\bibnamefont
  {Goold}}, \bibinfo {author} {\bibfnamefont {M.}~\bibnamefont {Huber}},
  \bibinfo {author} {\bibfnamefont {A.}~\bibnamefont {Riera}}, \bibinfo
  {author} {\bibfnamefont {L.}~\bibnamefont {del Rio}}, \ and\ \bibinfo
  {author} {\bibfnamefont {P.}~\bibnamefont {Skrzypczyk}},\ }\href@noop {}
  {\bibfield  {journal} {\bibinfo  {journal} {Journal of Physics A:
  Mathematical and Theoretical}\ }\textbf {\bibinfo {volume} {49}},\ \bibinfo
  {pages} {143001} (\bibinfo {year} {2016})}\BibitemShut {NoStop}%
\bibitem [{\citenamefont {Janzing}\ \emph {et~al.}(2000)\citenamefont
  {Janzing}, \citenamefont {Wocjan}, \citenamefont {Zeier}, \citenamefont
  {Geiss},\ and\ \citenamefont {Beth}}]{janzing2000thermodynamic}%
  \BibitemOpen
  \bibfield  {author} {\bibinfo {author} {\bibfnamefont {D.}~\bibnamefont
  {Janzing}}, \bibinfo {author} {\bibfnamefont {P.}~\bibnamefont {Wocjan}},
  \bibinfo {author} {\bibfnamefont {R.}~\bibnamefont {Zeier}}, \bibinfo
  {author} {\bibfnamefont {R.}~\bibnamefont {Geiss}}, \ and\ \bibinfo {author}
  {\bibfnamefont {T.}~\bibnamefont {Beth}},\ }\href@noop {} {\bibfield
  {journal} {\bibinfo  {journal} {International Journal of Theoretical
  Physics}\ }\textbf {\bibinfo {volume} {39}},\ \bibinfo {pages} {2717}
  (\bibinfo {year} {2000})}\BibitemShut {NoStop}%
\bibitem [{\citenamefont {Brandao}\ \emph {et~al.}(2013)\citenamefont
  {Brandao}, \citenamefont {Horodecki}, \citenamefont {Oppenheim},
  \citenamefont {Renes},\ and\ \citenamefont {Spekkens}}]{brandao2013resource}%
  \BibitemOpen
  \bibfield  {author} {\bibinfo {author} {\bibfnamefont {F.~G.}\ \bibnamefont
  {Brandao}}, \bibinfo {author} {\bibfnamefont {M.}~\bibnamefont {Horodecki}},
  \bibinfo {author} {\bibfnamefont {J.}~\bibnamefont {Oppenheim}}, \bibinfo
  {author} {\bibfnamefont {J.~M.}\ \bibnamefont {Renes}}, \ and\ \bibinfo
  {author} {\bibfnamefont {R.~W.}\ \bibnamefont {Spekkens}},\ }\href@noop {}
  {\bibfield  {journal} {\bibinfo  {journal} {Physical Review Letters}\
  }\textbf {\bibinfo {volume} {111}},\ \bibinfo {pages} {250404} (\bibinfo
  {year} {2013})}\BibitemShut {NoStop}%
\bibitem [{\citenamefont {DiVincenzo}\ \emph {et~al.}(1999)\citenamefont
  {DiVincenzo}, \citenamefont {Fuchs}, \citenamefont {Mabuchi}, \citenamefont
  {Smolin}, \citenamefont {Thapliyal},\ and\ \citenamefont
  {Uhlmann}}]{divincenzo1999entanglement}%
  \BibitemOpen
  \bibfield  {author} {\bibinfo {author} {\bibfnamefont {D.~P.}\ \bibnamefont
  {DiVincenzo}}, \bibinfo {author} {\bibfnamefont {C.~A.}\ \bibnamefont
  {Fuchs}}, \bibinfo {author} {\bibfnamefont {H.}~\bibnamefont {Mabuchi}},
  \bibinfo {author} {\bibfnamefont {J.~A.}\ \bibnamefont {Smolin}}, \bibinfo
  {author} {\bibfnamefont {A.}~\bibnamefont {Thapliyal}}, \ and\ \bibinfo
  {author} {\bibfnamefont {A.}~\bibnamefont {Uhlmann}},\ }in\ \href@noop {}
  {\emph {\bibinfo {booktitle} {Quantum Computing and Quantum
  Communications}}}\ (\bibinfo  {publisher} {Springer},\ \bibinfo {year}
  {1999})\ pp.\ \bibinfo {pages} {247--257}\BibitemShut {NoStop}%
\bibitem [{\citenamefont {Faist}\ \emph {et~al.}(2015)\citenamefont {Faist},
  \citenamefont {Oppenheim},\ and\ \citenamefont {Renner}}]{faist2015gibbs}%
  \BibitemOpen
  \bibfield  {author} {\bibinfo {author} {\bibfnamefont {P.}~\bibnamefont
  {Faist}}, \bibinfo {author} {\bibfnamefont {J.}~\bibnamefont {Oppenheim}}, \
  and\ \bibinfo {author} {\bibfnamefont {R.}~\bibnamefont {Renner}},\
  }\href@noop {} {\bibfield  {journal} {\bibinfo  {journal} {New Journal of
  Physics}\ }\textbf {\bibinfo {volume} {17}},\ \bibinfo {pages} {043003}
  (\bibinfo {year} {2015})}\BibitemShut {NoStop}%
\bibitem [{\citenamefont {Horodecki}\ and\ \citenamefont
  {Oppenheim}(2013)}]{horodecki2013fundamental}%
  \BibitemOpen
  \bibfield  {author} {\bibinfo {author} {\bibfnamefont {M.}~\bibnamefont
  {Horodecki}}\ and\ \bibinfo {author} {\bibfnamefont {J.}~\bibnamefont
  {Oppenheim}},\ }\href@noop {} {\bibfield  {journal} {\bibinfo  {journal}
  {Nature communications}\ }\textbf {\bibinfo {volume} {4}},\ \bibinfo {pages}
  {2059} (\bibinfo {year} {2013})}\BibitemShut {NoStop}%
\bibitem [{EPA()}]{EPAPS}%
  \BibitemOpen
  \href@noop {} {}\bibinfo {note} {See Supplemental Material at [EPAPS], which
  contains additional references
  \cite{horodecki1999reduction,King2003,HAYASHI,Fannes1973, Audenaert2007,
  Winter2016, vedral1998entanglement, winter2016tight}}\BibitemShut {NoStop}%
\bibitem [{\citenamefont {Narasimhachar}\ and\ \citenamefont
  {Gour}(2017)}]{narasimhachar2017resource}%
  \BibitemOpen
  \bibfield  {author} {\bibinfo {author} {\bibfnamefont {V.}~\bibnamefont
  {Narasimhachar}}\ and\ \bibinfo {author} {\bibfnamefont {G.}~\bibnamefont
  {Gour}},\ }\href@noop {} {\bibfield  {journal} {\bibinfo  {journal} {Physical
  Review A}\ }\textbf {\bibinfo {volume} {95}},\ \bibinfo {pages} {012313}
  (\bibinfo {year} {2017})}\BibitemShut {NoStop}%
\bibitem [{\citenamefont {Henderson}\ and\ \citenamefont
  {Vedral}(2001)}]{henderson2001classical}%
  \BibitemOpen
  \bibfield  {author} {\bibinfo {author} {\bibfnamefont {L.}~\bibnamefont
  {Henderson}}\ and\ \bibinfo {author} {\bibfnamefont {V.}~\bibnamefont
  {Vedral}},\ }\href@noop {} {\bibfield  {journal} {\bibinfo  {journal}
  {Journal of physics A: mathematical and general}\ }\textbf {\bibinfo {volume}
  {34}},\ \bibinfo {pages} {6899} (\bibinfo {year} {2001})}\BibitemShut
  {NoStop}%
\bibitem [{\citenamefont {Manzano}\ \emph {et~al.}(2018)\citenamefont
  {Manzano}, \citenamefont {Plastina},\ and\ \citenamefont
  {Zambrini}}]{PhysRevLett.121.120602}%
  \BibitemOpen
  \bibfield  {author} {\bibinfo {author} {\bibfnamefont {G.}~\bibnamefont
  {Manzano}}, \bibinfo {author} {\bibfnamefont {F.}~\bibnamefont {Plastina}}, \
  and\ \bibinfo {author} {\bibfnamefont {R.}~\bibnamefont {Zambrini}},\ }\href
  {\doibase 10.1103/PhysRevLett.121.120602} {\bibfield  {journal} {\bibinfo
  {journal} {Phys. Rev. Lett.}\ }\textbf {\bibinfo {volume} {121}},\ \bibinfo
  {pages} {120602} (\bibinfo {year} {2018})}\BibitemShut {NoStop}%
\bibitem [{\citenamefont {Devetak}\ and\ \citenamefont
  {Winter}(2004)}]{devetak2004distilling}%
  \BibitemOpen
  \bibfield  {author} {\bibinfo {author} {\bibfnamefont {I.}~\bibnamefont
  {Devetak}}\ and\ \bibinfo {author} {\bibfnamefont {A.}~\bibnamefont
  {Winter}},\ }\href@noop {} {\bibfield  {journal} {\bibinfo  {journal} {IEEE
  Transactions on Information Theory}\ }\textbf {\bibinfo {volume} {50}},\
  \bibinfo {pages} {3183} (\bibinfo {year} {2004})}\BibitemShut {NoStop}%
\bibitem [{\citenamefont {Oppenheim}\ \emph {et~al.}(2002)\citenamefont
  {Oppenheim}, \citenamefont {Horodecki}, \citenamefont {Horodecki},\ and\
  \citenamefont {Horodecki}}]{oppenheim2002thermodynamical}%
  \BibitemOpen
  \bibfield  {author} {\bibinfo {author} {\bibfnamefont {J.}~\bibnamefont
  {Oppenheim}}, \bibinfo {author} {\bibfnamefont {M.}~\bibnamefont
  {Horodecki}}, \bibinfo {author} {\bibfnamefont {P.}~\bibnamefont
  {Horodecki}}, \ and\ \bibinfo {author} {\bibfnamefont {R.}~\bibnamefont
  {Horodecki}},\ }\href@noop {} {\bibfield  {journal} {\bibinfo  {journal}
  {Physical Review Letters}\ }\textbf {\bibinfo {volume} {89}},\ \bibinfo
  {pages} {180402} (\bibinfo {year} {2002})}\BibitemShut {NoStop}%
\bibitem [{\citenamefont {Koashi}\ and\ \citenamefont
  {Winter}(2004)}]{koashi2004monogamy}%
  \BibitemOpen
  \bibfield  {author} {\bibinfo {author} {\bibfnamefont {M.}~\bibnamefont
  {Koashi}}\ and\ \bibinfo {author} {\bibfnamefont {A.}~\bibnamefont
  {Winter}},\ }\href@noop {} {\bibfield  {journal} {\bibinfo  {journal}
  {Physical Review A}\ }\textbf {\bibinfo {volume} {69}},\ \bibinfo {pages}
  {022309} (\bibinfo {year} {2004})}\BibitemShut {NoStop}%
\bibitem [{\citenamefont {Bennett}\ \emph {et~al.}(1996)\citenamefont
  {Bennett}, \citenamefont {Brassard}, \citenamefont {Popescu}, \citenamefont
  {Schumacher}, \citenamefont {Smolin},\ and\ \citenamefont
  {Wootters}}]{bennett1996purification}%
  \BibitemOpen
  \bibfield  {author} {\bibinfo {author} {\bibfnamefont {C.~H.}\ \bibnamefont
  {Bennett}}, \bibinfo {author} {\bibfnamefont {G.}~\bibnamefont {Brassard}},
  \bibinfo {author} {\bibfnamefont {S.}~\bibnamefont {Popescu}}, \bibinfo
  {author} {\bibfnamefont {B.}~\bibnamefont {Schumacher}}, \bibinfo {author}
  {\bibfnamefont {J.~A.}\ \bibnamefont {Smolin}}, \ and\ \bibinfo {author}
  {\bibfnamefont {W.~K.}\ \bibnamefont {Wootters}},\ }\href@noop {} {\bibfield
  {journal} {\bibinfo  {journal} {Physical Review Letters}\ }\textbf {\bibinfo
  {volume} {76}},\ \bibinfo {pages} {722} (\bibinfo {year} {1996})}\BibitemShut
  {NoStop}%
\bibitem [{\citenamefont {Hayden}\ \emph {et~al.}(2001)\citenamefont {Hayden},
  \citenamefont {Horodecki},\ and\ \citenamefont
  {Terhal}}]{hayden2001asymptotic}%
  \BibitemOpen
  \bibfield  {author} {\bibinfo {author} {\bibfnamefont {P.~M.}\ \bibnamefont
  {Hayden}}, \bibinfo {author} {\bibfnamefont {M.}~\bibnamefont {Horodecki}}, \
  and\ \bibinfo {author} {\bibfnamefont {B.~M.}\ \bibnamefont {Terhal}},\
  }\href@noop {} {\bibfield  {journal} {\bibinfo  {journal} {Journal of Physics
  A: Mathematical and General}\ }\textbf {\bibinfo {volume} {34}},\ \bibinfo
  {pages} {6891} (\bibinfo {year} {2001})}\BibitemShut {NoStop}%
\bibitem [{\citenamefont {Hastings}(2008)}]{hastings2008counterexample}%
  \BibitemOpen
  \bibfield  {author} {\bibinfo {author} {\bibfnamefont {M.~B.}\ \bibnamefont
  {Hastings}},\ }\href@noop {} {\bibfield  {journal} {\bibinfo  {journal}
  {arXiv preprint arXiv:0809.3972}\ } (\bibinfo {year} {2008})}\BibitemShut
  {NoStop}%
\bibitem [{\citenamefont {Devetak}\ and\ \citenamefont
  {Berger}(2001)}]{devetak2001low}%
  \BibitemOpen
  \bibfield  {author} {\bibinfo {author} {\bibfnamefont {I.}~\bibnamefont
  {Devetak}}\ and\ \bibinfo {author} {\bibfnamefont {T.}~\bibnamefont
  {Berger}},\ }\href@noop {} {\bibfield  {journal} {\bibinfo  {journal}
  {Physical Review Letters}\ }\textbf {\bibinfo {volume} {87}},\ \bibinfo
  {pages} {197901} (\bibinfo {year} {2001})}\BibitemShut {NoStop}%
\bibitem [{\citenamefont {Ollivier}\ and\ \citenamefont
  {Zurek}(2001)}]{ollivier2001quantum}%
  \BibitemOpen
  \bibfield  {author} {\bibinfo {author} {\bibfnamefont {H.}~\bibnamefont
  {Ollivier}}\ and\ \bibinfo {author} {\bibfnamefont {W.~H.}\ \bibnamefont
  {Zurek}},\ }\href@noop {} {\bibfield  {journal} {\bibinfo  {journal}
  {Physical Review Letters}\ }\textbf {\bibinfo {volume} {88}},\ \bibinfo
  {pages} {017901} (\bibinfo {year} {2001})}\BibitemShut {NoStop}%
\bibitem [{\citenamefont {Francica}\ \emph {et~al.}(2017)\citenamefont
  {Francica}, \citenamefont {Goold}, \citenamefont {Plastina},\ and\
  \citenamefont {Paternostro}}]{francica2017daemonic}%
  \BibitemOpen
  \bibfield  {author} {\bibinfo {author} {\bibfnamefont {G.}~\bibnamefont
  {Francica}}, \bibinfo {author} {\bibfnamefont {J.}~\bibnamefont {Goold}},
  \bibinfo {author} {\bibfnamefont {F.}~\bibnamefont {Plastina}}, \ and\
  \bibinfo {author} {\bibfnamefont {M.}~\bibnamefont {Paternostro}},\
  }\href@noop {} {\bibfield  {journal} {\bibinfo  {journal} {npj Quantum
  Information}\ }\textbf {\bibinfo {volume} {3}},\ \bibinfo {pages} {12}
  (\bibinfo {year} {2017})}\BibitemShut {NoStop}%
\bibitem [{\citenamefont {Zurek}(2003)}]{zurek2003quantum}%
  \BibitemOpen
  \bibfield  {author} {\bibinfo {author} {\bibfnamefont {W.~H.}\ \bibnamefont
  {Zurek}},\ }\href@noop {} {\bibfield  {journal} {\bibinfo  {journal}
  {Physical Review A}\ }\textbf {\bibinfo {volume} {67}},\ \bibinfo {pages}
  {012320} (\bibinfo {year} {2003})}\BibitemShut {NoStop}%
\bibitem [{\citenamefont {Liuzzo-Scorpo}\ \emph {et~al.}(2016)\citenamefont
  {Liuzzo-Scorpo}, \citenamefont {Correa}, \citenamefont {Schmidt},\ and\
  \citenamefont {Adesso}}]{liuzzo2016thermodynamics}%
  \BibitemOpen
  \bibfield  {author} {\bibinfo {author} {\bibfnamefont {P.}~\bibnamefont
  {Liuzzo-Scorpo}}, \bibinfo {author} {\bibfnamefont {L.~A.}\ \bibnamefont
  {Correa}}, \bibinfo {author} {\bibfnamefont {R.}~\bibnamefont {Schmidt}}, \
  and\ \bibinfo {author} {\bibfnamefont {G.}~\bibnamefont {Adesso}},\
  }\href@noop {} {\bibfield  {journal} {\bibinfo  {journal} {Entropy}\ }\textbf
  {\bibinfo {volume} {18}},\ \bibinfo {pages} {48} (\bibinfo {year}
  {2016})}\BibitemShut {NoStop}%
\bibitem [{\citenamefont {Adesso}\ \emph {et~al.}(2016)\citenamefont {Adesso},
  \citenamefont {Bromley},\ and\ \citenamefont
  {Cianciaruso}}]{adesso2016measures}%
  \BibitemOpen
  \bibfield  {author} {\bibinfo {author} {\bibfnamefont {G.}~\bibnamefont
  {Adesso}}, \bibinfo {author} {\bibfnamefont {T.~R.}\ \bibnamefont {Bromley}},
  \ and\ \bibinfo {author} {\bibfnamefont {M.}~\bibnamefont {Cianciaruso}},\
  }\href@noop {} {\bibfield  {journal} {\bibinfo  {journal} {Journal of Physics
  A: Mathematical and Theoretical}\ }\textbf {\bibinfo {volume} {49}},\
  \bibinfo {pages} {473001} (\bibinfo {year} {2016})}\BibitemShut {NoStop}%
\bibitem [{\citenamefont {Chitambar}\ \emph {et~al.}(2016)\citenamefont
  {Chitambar}, \citenamefont {Streltsov}, \citenamefont {Rana}, \citenamefont
  {Bera}, \citenamefont {Adesso},\ and\ \citenamefont
  {Lewenstein}}]{chitambar2016assisted}%
  \BibitemOpen
  \bibfield  {author} {\bibinfo {author} {\bibfnamefont {E.}~\bibnamefont
  {Chitambar}}, \bibinfo {author} {\bibfnamefont {A.}~\bibnamefont
  {Streltsov}}, \bibinfo {author} {\bibfnamefont {S.}~\bibnamefont {Rana}},
  \bibinfo {author} {\bibfnamefont {M.}~\bibnamefont {Bera}}, \bibinfo {author}
  {\bibfnamefont {G.}~\bibnamefont {Adesso}}, \ and\ \bibinfo {author}
  {\bibfnamefont {M.}~\bibnamefont {Lewenstein}},\ }\href@noop {} {\bibfield
  {journal} {\bibinfo  {journal} {Physical Review Letters}\ }\textbf {\bibinfo
  {volume} {116}},\ \bibinfo {pages} {070402} (\bibinfo {year}
  {2016})}\BibitemShut {NoStop}%
\bibitem [{\citenamefont {Horodecki}\ \emph {et~al.}(2003)\citenamefont
  {Horodecki}, \citenamefont {Horodecki}, \citenamefont {Horodecki},
  \citenamefont {Horodecki}, \citenamefont {Oppenheim}, \citenamefont {Sen},
  \citenamefont {Sen} \emph {et~al.}}]{horodecki2003local}%
  \BibitemOpen
  \bibfield  {author} {\bibinfo {author} {\bibfnamefont {M.}~\bibnamefont
  {Horodecki}}, \bibinfo {author} {\bibfnamefont {K.}~\bibnamefont
  {Horodecki}}, \bibinfo {author} {\bibfnamefont {P.}~\bibnamefont
  {Horodecki}}, \bibinfo {author} {\bibfnamefont {R.}~\bibnamefont
  {Horodecki}}, \bibinfo {author} {\bibfnamefont {J.}~\bibnamefont
  {Oppenheim}}, \bibinfo {author} {\bibfnamefont {A.}~\bibnamefont {Sen}},
  \bibinfo {author} {\bibfnamefont {U.}~\bibnamefont {Sen}},  \emph {et~al.},\
  }\href@noop {} {\bibfield  {journal} {\bibinfo  {journal} {Physical review
  letters}\ }\textbf {\bibinfo {volume} {90}},\ \bibinfo {pages} {100402}
  (\bibinfo {year} {2003})}\BibitemShut {NoStop}%
\bibitem [{\citenamefont {Devetak}(2005)}]{devetak2005distillation}%
  \BibitemOpen
  \bibfield  {author} {\bibinfo {author} {\bibfnamefont {I.}~\bibnamefont
  {Devetak}},\ }\href@noop {} {\bibfield  {journal} {\bibinfo  {journal}
  {Physical Review A}\ }\textbf {\bibinfo {volume} {71}},\ \bibinfo {pages}
  {062303} (\bibinfo {year} {2005})}\BibitemShut {NoStop}%
\bibitem [{\citenamefont {Regula}\ \emph {et~al.}(2018)\citenamefont {Regula},
  \citenamefont {Lami},\ and\ \citenamefont {Streltsov}}]{bartosz-myself-alex}%
  \BibitemOpen
  \bibfield  {author} {\bibinfo {author} {\bibfnamefont {B.}~\bibnamefont
  {Regula}}, \bibinfo {author} {\bibfnamefont {L.}~\bibnamefont {Lami}}, \ and\
  \bibinfo {author} {\bibfnamefont {A.}~\bibnamefont {Streltsov}},\ }\href@noop
  {} {\bibfield  {journal} {\bibinfo  {journal} {Physical Review A}\ }\textbf
  {\bibinfo {volume} {98}},\ \bibinfo {pages} {052329} (\bibinfo {year}
  {2018})}\BibitemShut {NoStop}%
\bibitem [{\citenamefont {Streltsov}\ \emph
  {et~al.}(2017{\natexlab{a}})\citenamefont {Streltsov}, \citenamefont {Rana},
  \citenamefont {Bera},\ and\ \citenamefont
  {Lewenstein}}]{streltsov2017towards}%
  \BibitemOpen
  \bibfield  {author} {\bibinfo {author} {\bibfnamefont {A.}~\bibnamefont
  {Streltsov}}, \bibinfo {author} {\bibfnamefont {S.}~\bibnamefont {Rana}},
  \bibinfo {author} {\bibfnamefont {M.~N.}\ \bibnamefont {Bera}}, \ and\
  \bibinfo {author} {\bibfnamefont {M.}~\bibnamefont {Lewenstein}},\
  }\href@noop {} {\bibfield  {journal} {\bibinfo  {journal} {Physical Review
  X}\ }\textbf {\bibinfo {volume} {7}},\ \bibinfo {pages} {011024} (\bibinfo
  {year} {2017}{\natexlab{a}})}\BibitemShut {NoStop}%
\bibitem [{\citenamefont {Streltsov}\ \emph
  {et~al.}(2017{\natexlab{b}})\citenamefont {Streltsov}, \citenamefont
  {Adesso},\ and\ \citenamefont {Plenio}}]{streltsov2017colloquium}%
  \BibitemOpen
  \bibfield  {author} {\bibinfo {author} {\bibfnamefont {A.}~\bibnamefont
  {Streltsov}}, \bibinfo {author} {\bibfnamefont {G.}~\bibnamefont {Adesso}}, \
  and\ \bibinfo {author} {\bibfnamefont {M.~B.}\ \bibnamefont {Plenio}},\
  }\href@noop {} {\bibfield  {journal} {\bibinfo  {journal} {Reviews of Modern
  Physics}\ }\textbf {\bibinfo {volume} {89}},\ \bibinfo {pages} {041003}
  (\bibinfo {year} {2017}{\natexlab{b}})}\BibitemShut {NoStop}%
\bibitem [{\citenamefont {Baumgratz}\ \emph {et~al.}(2014)\citenamefont
  {Baumgratz}, \citenamefont {Cramer},\ and\ \citenamefont
  {Plenio}}]{baumgratz2014quantifying}%
  \BibitemOpen
  \bibfield  {author} {\bibinfo {author} {\bibfnamefont {T.}~\bibnamefont
  {Baumgratz}}, \bibinfo {author} {\bibfnamefont {M.}~\bibnamefont {Cramer}}, \
  and\ \bibinfo {author} {\bibfnamefont {M.~B.}\ \bibnamefont {Plenio}},\
  }\href@noop {} {\bibfield  {journal} {\bibinfo  {journal} {Physical review
  letters}\ }\textbf {\bibinfo {volume} {113}},\ \bibinfo {pages} {140401}
  (\bibinfo {year} {2014})}\BibitemShut {NoStop}%
\bibitem [{\citenamefont {Winter}\ and\ \citenamefont
  {Yang}(2016)}]{winter2016operational}%
  \BibitemOpen
  \bibfield  {author} {\bibinfo {author} {\bibfnamefont {A.}~\bibnamefont
  {Winter}}\ and\ \bibinfo {author} {\bibfnamefont {D.}~\bibnamefont {Yang}},\
  }\href@noop {} {\bibfield  {journal} {\bibinfo  {journal} {Physical review
  letters}\ }\textbf {\bibinfo {volume} {116}},\ \bibinfo {pages} {120404}
  (\bibinfo {year} {2016})}\BibitemShut {NoStop}%
\bibitem [{\citenamefont {Szilard}(1929)}]{szilard1929entropieverminderung}%
  \BibitemOpen
  \bibfield  {author} {\bibinfo {author} {\bibfnamefont {L.}~\bibnamefont
  {Szilard}},\ }\href@noop {} {\bibfield  {journal} {\bibinfo  {journal}
  {Zeitschrift f{\"u}r Physik}\ }\textbf {\bibinfo {volume} {53}},\ \bibinfo
  {pages} {840} (\bibinfo {year} {1929})}\BibitemShut {NoStop}%
\bibitem [{\citenamefont {Landauer}(1961)}]{landauer1961irreversibility}%
  \BibitemOpen
  \bibfield  {author} {\bibinfo {author} {\bibfnamefont {R.}~\bibnamefont
  {Landauer}},\ }\href@noop {} {\bibfield  {journal} {\bibinfo  {journal} {IBM
  journal of research and development}\ }\textbf {\bibinfo {volume} {5}},\
  \bibinfo {pages} {183} (\bibinfo {year} {1961})}\BibitemShut {NoStop}%
\bibitem [{\citenamefont {Kim}\ \emph {et~al.}(2011)\citenamefont {Kim},
  \citenamefont {Sagawa}, \citenamefont {De~Liberato},\ and\ \citenamefont
  {Ueda}}]{kim2011quantum}%
  \BibitemOpen
  \bibfield  {author} {\bibinfo {author} {\bibfnamefont {S.~W.}\ \bibnamefont
  {Kim}}, \bibinfo {author} {\bibfnamefont {T.}~\bibnamefont {Sagawa}},
  \bibinfo {author} {\bibfnamefont {S.}~\bibnamefont {De~Liberato}}, \ and\
  \bibinfo {author} {\bibfnamefont {M.}~\bibnamefont {Ueda}},\ }\href@noop {}
  {\bibfield  {journal} {\bibinfo  {journal} {Physical Review Letters}\
  }\textbf {\bibinfo {volume} {106}},\ \bibinfo {pages} {070401} (\bibinfo
  {year} {2011})}\BibitemShut {NoStop}%
\bibitem [{\citenamefont {Mohammady}\ and\ \citenamefont
  {Anders}(2017)}]{mohammady2017quantum}%
  \BibitemOpen
  \bibfield  {author} {\bibinfo {author} {\bibfnamefont {M.~H.}\ \bibnamefont
  {Mohammady}}\ and\ \bibinfo {author} {\bibfnamefont {J.}~\bibnamefont
  {Anders}},\ }\href@noop {} {\bibfield  {journal} {\bibinfo  {journal} {New
  Journal of Physics}\ }\textbf {\bibinfo {volume} {19}},\ \bibinfo {pages}
  {113026} (\bibinfo {year} {2017})}\BibitemShut {NoStop}%
\bibitem [{\citenamefont {Cottet}\ \emph {et~al.}(2017)\citenamefont {Cottet},
  \citenamefont {Jezouin}, \citenamefont {Bretheau}, \citenamefont
  {Campagne-Ibarcq}, \citenamefont {Ficheux}, \citenamefont {Anders},
  \citenamefont {Auff{\`e}ves}, \citenamefont {Azouit}, \citenamefont
  {Rouchon},\ and\ \citenamefont {Huard}}]{cottet2017observing}%
  \BibitemOpen
  \bibfield  {author} {\bibinfo {author} {\bibfnamefont {N.}~\bibnamefont
  {Cottet}}, \bibinfo {author} {\bibfnamefont {S.}~\bibnamefont {Jezouin}},
  \bibinfo {author} {\bibfnamefont {L.}~\bibnamefont {Bretheau}}, \bibinfo
  {author} {\bibfnamefont {P.}~\bibnamefont {Campagne-Ibarcq}}, \bibinfo
  {author} {\bibfnamefont {Q.}~\bibnamefont {Ficheux}}, \bibinfo {author}
  {\bibfnamefont {J.}~\bibnamefont {Anders}}, \bibinfo {author} {\bibfnamefont
  {A.}~\bibnamefont {Auff{\`e}ves}}, \bibinfo {author} {\bibfnamefont
  {R.}~\bibnamefont {Azouit}}, \bibinfo {author} {\bibfnamefont
  {P.}~\bibnamefont {Rouchon}}, \ and\ \bibinfo {author} {\bibfnamefont
  {B.}~\bibnamefont {Huard}},\ }\href@noop {} {\bibfield  {journal} {\bibinfo
  {journal} {Proceedings of the National Academy of Sciences}\ }\textbf
  {\bibinfo {volume} {114}},\ \bibinfo {pages} {7561} (\bibinfo {year}
  {2017})}\BibitemShut {NoStop}%
\bibitem [{\citenamefont {Park}\ \emph {et~al.}(2013)\citenamefont {Park},
  \citenamefont {Kim}, \citenamefont {Sagawa},\ and\ \citenamefont
  {Kim}}]{park2013heat}%
  \BibitemOpen
  \bibfield  {author} {\bibinfo {author} {\bibfnamefont {J.~J.}\ \bibnamefont
  {Park}}, \bibinfo {author} {\bibfnamefont {K.-H.}\ \bibnamefont {Kim}},
  \bibinfo {author} {\bibfnamefont {T.}~\bibnamefont {Sagawa}}, \ and\ \bibinfo
  {author} {\bibfnamefont {S.~W.}\ \bibnamefont {Kim}},\ }\href@noop {}
  {\bibfield  {journal} {\bibinfo  {journal} {Physical Review Letters}\
  }\textbf {\bibinfo {volume} {111}},\ \bibinfo {pages} {230402} (\bibinfo
  {year} {2013})}\BibitemShut {NoStop}%
\bibitem [{\citenamefont {Zurek}(1986)}]{zurek1986maxwell}%
  \BibitemOpen
  \bibfield  {author} {\bibinfo {author} {\bibfnamefont {W.~H.}\ \bibnamefont
  {Zurek}},\ }in\ \href@noop {} {\emph {\bibinfo {booktitle} {Frontiers of
  nonequilibrium statistical physics}}}\ (\bibinfo  {publisher} {Springer},\
  \bibinfo {year} {1986})\ pp.\ \bibinfo {pages} {151--161}\BibitemShut
  {NoStop}%
\bibitem [{\citenamefont {Reeb}\ and\ \citenamefont
  {Wolf}(2014)}]{reeb2014improved}%
  \BibitemOpen
  \bibfield  {author} {\bibinfo {author} {\bibfnamefont {D.}~\bibnamefont
  {Reeb}}\ and\ \bibinfo {author} {\bibfnamefont {M.~M.}\ \bibnamefont
  {Wolf}},\ }\href@noop {} {\bibfield  {journal} {\bibinfo  {journal} {New
  Journal of Physics}\ }\textbf {\bibinfo {volume} {16}},\ \bibinfo {pages}
  {103011} (\bibinfo {year} {2014})}\BibitemShut {NoStop}%
\bibitem [{\citenamefont {Sagawa}\ and\ \citenamefont
  {Ueda}(2008)}]{sagawa2008second}%
  \BibitemOpen
  \bibfield  {author} {\bibinfo {author} {\bibfnamefont {T.}~\bibnamefont
  {Sagawa}}\ and\ \bibinfo {author} {\bibfnamefont {M.}~\bibnamefont {Ueda}},\
  }\href@noop {} {\bibfield  {journal} {\bibinfo  {journal} {Physical Review
  Letters}\ }\textbf {\bibinfo {volume} {100}},\ \bibinfo {pages} {080403}
  (\bibinfo {year} {2008})}\BibitemShut {NoStop}%
\bibitem [{\citenamefont {Narasimhachar}\ \emph {et~al.}(2018)\citenamefont
  {Narasimhachar}, \citenamefont {Thompson}, \citenamefont {Ma}, \citenamefont
  {Gour},\ and\ \citenamefont {Gu}}]{narasimhachar2018quantifying}%
  \BibitemOpen
  \bibfield  {author} {\bibinfo {author} {\bibfnamefont {V.}~\bibnamefont
  {Narasimhachar}}, \bibinfo {author} {\bibfnamefont {J.}~\bibnamefont
  {Thompson}}, \bibinfo {author} {\bibfnamefont {J.}~\bibnamefont {Ma}},
  \bibinfo {author} {\bibfnamefont {G.}~\bibnamefont {Gour}}, \ and\ \bibinfo
  {author} {\bibfnamefont {M.}~\bibnamefont {Gu}},\ }\href@noop {} {\bibfield
  {journal} {\bibinfo  {journal} {arXiv preprint arXiv:1806.00025}\ } (\bibinfo
  {year} {2018})}\BibitemShut {NoStop}%
\bibitem [{\citenamefont {Horodecki}\ and\ \citenamefont
  {Horodecki}(1999)}]{horodecki1999reduction}%
  \BibitemOpen
  \bibfield  {author} {\bibinfo {author} {\bibfnamefont {M.}~\bibnamefont
  {Horodecki}}\ and\ \bibinfo {author} {\bibfnamefont {P.}~\bibnamefont
  {Horodecki}},\ }\href@noop {} {\bibfield  {journal} {\bibinfo  {journal}
  {Physical Review A}\ }\textbf {\bibinfo {volume} {59}},\ \bibinfo {pages}
  {4206} (\bibinfo {year} {1999})}\BibitemShut {NoStop}%
\bibitem [{\citenamefont {King}(2003)}]{King2003}%
  \BibitemOpen
  \bibfield  {author} {\bibinfo {author} {\bibfnamefont {C.}~\bibnamefont
  {King}},\ }\href {\doibase 10.1109/TIT.2002.806153} {\bibfield  {journal}
  {\bibinfo  {journal} {IEEE Trans. Inf. Theory}\ }\textbf {\bibinfo {volume}
  {49}},\ \bibinfo {pages} {221} (\bibinfo {year} {2003})}\BibitemShut
  {NoStop}%
\bibitem [{\citenamefont {Hayashi}(2016)}]{HAYASHI}%
  \BibitemOpen
  \bibfield  {author} {\bibinfo {author} {\bibfnamefont {M.}~\bibnamefont
  {Hayashi}},\ }\href@noop {} {\emph {\bibinfo {title} {Quantum Information
  Theory: Mathematical Foundation}}},\ Graduate Texts in Physics\ (\bibinfo
  {publisher} {Springer Berlin Heidelberg},\ \bibinfo {year}
  {2016})\BibitemShut {NoStop}%
\bibitem [{\citenamefont {Fannes}(1973)}]{Fannes1973}%
  \BibitemOpen
  \bibfield  {author} {\bibinfo {author} {\bibfnamefont {M.}~\bibnamefont
  {Fannes}},\ }\href {\doibase 10.1007/BF01646490} {\bibfield  {journal}
  {\bibinfo  {journal} {Commun. Math. Phys.}\ }\textbf {\bibinfo {volume}
  {31}},\ \bibinfo {pages} {291} (\bibinfo {year} {1973})}\BibitemShut
  {NoStop}%
\bibitem [{\citenamefont {Audenaert}(2007)}]{Audenaert2007}%
  \BibitemOpen
  \bibfield  {author} {\bibinfo {author} {\bibfnamefont {K.}~\bibnamefont
  {Audenaert}},\ }\href@noop {} {\bibfield  {journal} {\bibinfo  {journal} {J.
  Phys. A}\ }\textbf {\bibinfo {volume} {40}},\ \bibinfo {pages} {8127}
  (\bibinfo {year} {2007})}\BibitemShut {NoStop}%
\bibitem [{\citenamefont {Winter}(2016{\natexlab{a}})}]{Winter2016}%
  \BibitemOpen
  \bibfield  {author} {\bibinfo {author} {\bibfnamefont {A.}~\bibnamefont
  {Winter}},\ }\href {\doibase 10.1007/s00220-016-2609-8} {\bibfield  {journal}
  {\bibinfo  {journal} {Commun. Math. Phys.}\ }\textbf {\bibinfo {volume}
  {347}},\ \bibinfo {pages} {291} (\bibinfo {year}
  {2016}{\natexlab{a}})}\BibitemShut {NoStop}%
\bibitem [{\citenamefont {Vedral}\ and\ \citenamefont
  {Plenio}(1998)}]{vedral1998entanglement}%
  \BibitemOpen
  \bibfield  {author} {\bibinfo {author} {\bibfnamefont {V.}~\bibnamefont
  {Vedral}}\ and\ \bibinfo {author} {\bibfnamefont {M.~B.}\ \bibnamefont
  {Plenio}},\ }\href@noop {} {\bibfield  {journal} {\bibinfo  {journal}
  {Physical Review A}\ }\textbf {\bibinfo {volume} {57}},\ \bibinfo {pages}
  {1619} (\bibinfo {year} {1998})}\BibitemShut {NoStop}%
\bibitem [{\citenamefont {Winter}(2016{\natexlab{b}})}]{winter2016tight}%
  \BibitemOpen
  \bibfield  {author} {\bibinfo {author} {\bibfnamefont {A.}~\bibnamefont
  {Winter}},\ }\href@noop {} {\bibfield  {journal} {\bibinfo  {journal}
  {Communications in Mathematical Physics}\ }\textbf {\bibinfo {volume}
  {347}},\ \bibinfo {pages} {291} (\bibinfo {year}
  {2016}{\natexlab{b}})}\BibitemShut {NoStop}%
\end{thebibliography}


%\end{document}

%%%%%%%%%%%%%%%%%%%%%%%%%%%%%%%%%%%%%%%%%%%%%%%%%%%%%%%%%%%%%%%%%%%%%%%%%%%%%%%%%%%%%%%%%%%%%%%%%%%%%%%%%%%%%%%%%%%%

\clearpage

\onecolumngrid
\appendix
 \setcounter{page}{1}
\begin{center}
\vspace*{\baselineskip}
{\textbf{\large Assisted Work Distillation: Supplemental Material}}\\
\end{center}
%\twocolumngrid

%\renewcommand{\theequation}{S\arabic{equation}}
%\setcounter{equation}{0}
%\setcounter{figure}{0}
%\setcounter{table}{0}
%\setcounter{section}{0}
%\setcounter{page}{1}
%\makeatletter

\section{Distillable work under thermal operations} \label{worksupreme}

Here we present an explicit proof of Eq.~\eqref{workisfreeenergy} in the main text, formalizing the connection between the distillable work under TO and the relative entropy of athermality. In what follows, will drop the system subscript $B$ for simplicity.

We start by recalling the main result of~\cite{brandao2013resource}: given an initial state $\omega$ of a system with Hamiltonian $H$, and a target state $\omega'$ of a system with Hamiltonian $H'$, the asymptotically achievable rates $R$ in a state transformation $\omega^{\otimes n}\to(\omega')^{\otimes [Rn]}$ operated by TO at background inverse temperature $\beta$ with vanishing error must satisfy
\begin{equation}
S\left(\omega^{\otimes n} \Big\| \gamma^{\otimes n}\right) \geq S\left((\omega')^{\otimes [Rn]} \Big\| (\gamma')^{\otimes [Rn]}\right) ,
\end{equation}
where $\gamma\coloneqq Z^{-1} e^{-\beta H}$ and $\gamma'\coloneqq (Z')^{-1} e^{-\beta H'}$ are the Gibbs states of the input and output systems, respectively. In our case, $\omega^{\otimes n} = \rho^{\otimes n} \otimes \ketbra{0}_P^{\otimes [Rn]}$ and $(\omega')^{\otimes [Rn]} = \ketbra{1}_P^{\otimes [Rn]}$, so that
\begin{align*}
0 &\leq S\left( \rho^{\otimes n} \otimes \ketbra{0}_P^{\otimes [Rn]} \Big\| \gamma^{\otimes n} \otimes \gamma_P^{\otimes [Rn]} \right) - S\left( \ketbra{1}_P^{\otimes [Rn]} \Big\| \gamma_P^{\otimes [Rn]} \right) \\[.5ex]
&= n S(\rho\|\gamma) + [Rn]\, S\left(\ketbra{0}_P\|\gamma_P\right) - [Rn]\, S\left(\ketbra{1}_P\|\gamma_P\right) \\
&= n S(\rho\|\gamma) + [Rn]\, \log \left( 1+e^{-\beta E}\right) - [Rn]\, \log \left( \frac{1+e^{-\beta E}}{e^{-\beta E}}\right) \\
&= n S(\rho\|\gamma) - [Rn]\, \beta E\, ,
\end{align*}
where we observed that for the qubit battery one has
\begin{equation}
\gamma_P = \frac{1}{1+e^{-\beta E}} \begin{pmatrix} 1 & 0 \\ 0 & e^{-\beta E} \end{pmatrix} .
\label{Gibbs battery}
\end{equation}
We deduce that
\begin{equation*}
RE\leq \frac{1}{\beta} S(\rho\|\gamma)\, .
\end{equation*}
Taking the supremum as dictated by Eq.~\eqref{distillablework} yields Eq.~\eqref{workisfreeenergy}, as claimed.

\section{Bipartite quantum thermal states} \label{bipart}
\noindent As we have seen in the main text, the assisted framework for work distillation allows Alice to perform any given POVM $\{\Pi_{A,i}\}_i$. It is natural to ask, under what conditions the assistance by Alice is a valuable resource that helps to distil more work. In this section we show that this is indeed the case whenever the bipartite state Alice and Bob share is not a product state. In other words, any state which cannot be created for free via the allowed operations constitutes a resource for extracting work on Bob's side.

\begin{proposition} \label{assistance helps}
Whenever $\rho_{AB}\neq \rho_A\otimes \rho_B$ is not factorized, we have that
\begin{equation}
    W_a^{B|A}(\rho_{AB}) > \frac{1}{\beta} S(\rho_B\|\gamma_B)\, .
\end{equation}
In other words, the assistance by Alice allows to extract more work than Bob could in the unassisted setting.
\end{proposition}

Before we present a proof of the above result, it may be useful to recall an elementary lemma.

\begin{lemma} \label{lemmino}
A bipartite quantum state $\rho_{AB}$ is factorized iff for all $\Pi_A\geq 0$ the operator $\Tr_A\left[\Pi_A\otimes \id_B\, \rho_{AB}\right]$ is proportional to a fixed state $\sigma_B$ (independent of $\Pi_A$).
\end{lemma}

\begin{proof}
Since any operator can be written as a complex linear combination of at most four positive operators, we deduce that for all operators $N_A$ we have that $\Tr_A\left[N_A\otimes \id_B\, \rho_{AB}\right] = c(N)\, \sigma_B$, where $c(N)$ is a complex scalar. Choosing $N_A = \ketbraa{i}{j}_A$ for some basis $\ket{i}_A$ of the Hilbert space on $A$, and then summing over $i,j$, one obtains
\begin{align*}
    \rho_{AB} &= \sum_{i,j} \ketbraa{i}{j}_A\otimes \Tr_A\left[\ketbraa{i}{j}_A \otimes \id_B\, \rho_{AB}\right] \\
    &= \sum_{i,j} \ketbraa{i}{j}_A\otimes c_{ij} \sigma_B \\
    &= \left( \sum_{i,j} c_{ij} \ketbraa{i}{j} \right)_A \otimes \sigma_B\, ,
\end{align*}
implying that $\rho_{AB}=\rho_A\otimes \rho_B$ is factorized.
\end{proof}

We are now ready to prove the main result of this section.

\begin{proof}[Proof of Proposition~\ref{assistance helps}]
Upon measuring her subsystem, Alice will leave Bob in the state $\tilde{\rho}_{B,i}\propto \Tr_A\left[\Pi_{A,i}\otimes \id_B\, \rho_{AB}\right]$ of~\eqref{rho tilde} with probability $p_i = \Tr \left[\rho_{A} \Pi_{A,i}\right]$. By Lemma~\ref{lemmino}, we know that if $\rho_{AB}\neq \rho_A\otimes \rho_B$ there are two positive operators $\Pi_{A,1}, \Pi_{A,2}\geq 0$ such that $\tilde{\rho}_{B,1} \neq \tilde{\rho}_{B,2}$. Up to rescaling them, we can make sure that $\Pi_{A,1}+\Pi_{A,2}\leq \id_A$, so that there exists a valid POVM that includes both $\Pi_{A,1}$ and $\Pi_{A,2}$. Invoking the strict concavity of the entropy, it is then elementary to establish that
\begin{align*}
    \beta\, W_a^{B|A}(\rho_{AB}) &\geq \sum_i p_i S(\tilde{\rho}_{B,i} \| \gamma_B) \\
    &= - \sum_i p_i S(\tilde{\rho}_{B,i}) - \sum_i p_i \Tr \left[\tilde{\rho}_{B,i} \log \gamma_B\right] \\
    &> - S\left( \sum_i p_i \tilde{\rho}_{B,i}\right) - \Tr \left[ \left(\sum_i p_i \tilde{\rho}_{B,i}\right) \log\gamma_B \right] \\
    &= - S(\rho_{B}) - \Tr\left[ \rho_{B} \log \gamma_B\right] \\
    &= S(\rho_B \| \gamma_B)\, ,
\end{align*}
which concludes the proof.
\end{proof}

\section{Work of assistance as a function of the Henderson--Vedral measure} \label{workasshen}
\noindent Suppose Alice performs a general quantum operation $\{\Pi_{A,i}\}_i$, causing Bob to receive a state $\tilde{\rho}_{B,i}$, with the associated probabilities $\Tr \rho_A \Pi_{A,i}=p_i$. She can use her classical communication channel to choose a measurement so that Bob's extractable work is maximized,
\begin{align}\label{nonpure}
W_a^{B|A}(\rho_{AB})=&\max_{\{\Pi_{A,i}\}}\frac{1}{\beta}\sum_{i}p_i  S\left(\tilde{\rho}_{B,i}||\gamma_B\right),\\
=&\max_{\{\Pi_{A,i}\}}\left(\langle E \rangle_{\tilde{\rho}_{B,i}}+\frac{1}{\beta}\left[\ln Z_B -\sum_{i}p_{i}S(\tilde{\rho}_{B,i})\right]\right).
\end{align}
The work of assistance can also be written in the following way,
\begin{align}\label{workass0}
W_a^{B|A}(\rho_{AB})=&\max_{\{\Pi_{A,i}\}}\frac{1}{\beta}\sum_i p_i S(\tilde{\rho}_{B,i}||\gamma_B)\nonumber\\
=&\max_{\{\Pi_{A,i}\}}\frac{1}{\beta}\sum_{i}p_i\Tr\left[\tilde{\rho}_{B,i}\left(\log\tilde{\rho}_{B,i}-\log \gamma_B\right)\right]\nonumber\\
=&\max_{\{\Pi_{A,i}\}}\frac{1}{\beta}\sum_{i}p_i\left[-S(\tilde{\rho}_{B,i})-\Tr\tilde{\rho}_{B,i}\log\gamma_B\right]\nonumber\\
=&\max_{\{\Pi_{A,i}\}}\frac{1}{\beta}\left[-\sum_ip_iS(\tilde{\rho}_{B,i})-\Tr\rho_B\log\gamma_B\right],
\end{align}
as $\sum p_i\tilde{\rho}_{B,i}=\rho_B$. Therefore the work of assistance refers to Alice attempting to minimize the local entropy on Bob's side via her measurement,
\begin{align}\label{workass5}
W_a^{B|A}(\rho_{AB})=-\frac{1}{\beta}\left[\Tr\rho_B\log \gamma_B+\min_{\{\Pi_{A,i}\}_i}\sum_i p_iS(\tilde{\rho}_{B,i})\right].
\end{align}
We can lower bound this value using the convexity of the quantum relative entropy
\begin{align}
W_a^{B|A}(\rho_{AB})\ge\frac{1}{\beta}S(\rho_B||\gamma_B),
\end{align}
which is also additive. In order to further investigate the work of assistance we can employ a measure of classical correlations from~\cite{henderson2001classical} defined as
\begin{align}\label{hender}
J^{\rightarrow}(\rho_{AB})&=\max_{\{\Pi_{A,i}\}_i}\left[S(\rho_B)-\sum_i p_iS(\tilde{\rho}_{B,i})\right],
\end{align}
where, like the work of assistance~\eqref{workass0}, the maximization is taken over all the measurements $\{\Pi_{A,i}\}$ applied on Alice's subsystem. Substituting this into~\eqref{workass1} allows us to write the work of assistance as
\begin{align}\label{workass6}
W_a^{B|A}(\rho_{AB})=&-\frac{1}{\beta}\Tr\left[\rho_B\log\gamma_B\right]-\frac{1}{\beta}S(\rho_B)+\frac{1}{\beta}	J^{\rightarrow}(\rho_{AB}),\nonumber\\
=&\frac{1}{\beta}\left(S(\rho_B||\gamma_B)+	J^{\rightarrow}(\rho_{AB})\right).
\end{align}

\section{The regularized work of assistance is the maximum distillable work with one-way communication}
\label{W a infty optimal 1-way}

Here we argue that the regularized work of assistance $W_{a,\infty}^{B|A}$ of Eq.~\eqref{regularized} coincides with the optimal work distillation rate when only Alice $\to$ Bob classical communication is allowed. In this case, it makes no difference whether Bob has access to all Gibbs-preserving operations or only to thermal ones. Therefore, our claim also implies that the work of collaboration $W_c^{B|A}$ of Eq.~\eqref{collab1} coincides with $W_{a,\infty}^{B|A}$ when the operations on Bob's side are required to be thermal, i.e.\ $\mathcal{O}=\mathrm{TO}$.

Start by considering a rate $R$ that is achievable for a fixed state $\rho$ of $AB$ in the sense of Eq.~\eqref{collab1}, where we assume from now on that $\mathcal{O}=\mathrm{TO}$. This means that there is a sequence $\{\Lambda_n\}_n$ of operations on $A:BP$ obtained by concatenating arbitrary quantum instruments on Alice, classical communication Alice $\to$ Bob, and TO on Bob's side (which includes the battery $P$), such that
\begin{equation*}
    \Lambda_n\Big(\rho_{AB}^{\otimes n}\otimes \ketbra{0}_P^{\otimes Rn}\Big) = \ketbra{1}_P^{\otimes[Rn]} + \delta_n\, ,
\end{equation*}
where the above equation defines a sequence of ``remainder terms'' $\{\delta_n\}_n$ that: (i) are traceless, i.e.\ such that $\Tr \delta_n=0$; (ii) satisfy $\epsilon_n\coloneqq \|\delta_n\|_1\to 0$ as $n\to \infty$. By definition of work of assistance, we can write
\begin{equation}
\begin{aligned}
    \frac{1}{n} W_a^{B|A}\left(\rho_{AB}^{\otimes n}\right) &\geq \frac{1}{n\beta} S\left( \ketbra{1}_P^{\otimes [Rn]} +\delta_n \Big\| \gamma_P^{\otimes n}\right) \\
    &= -\frac{1}{n\beta} S\left( \ketbra{1}_P^{\otimes [Rn]} + \delta_n\right) - \frac{1}{n\beta} \Tr\left[ \left( \ketbra{1}_P^{\otimes [Rn]} + \delta_n \right) \log \left( \gamma_P^{\otimes [Rn]}\right)\right] \\
    &= -\frac{1}{n\beta} S\left( \ketbra{1}_P^{\otimes [Rn]} + \delta_n\right) - \frac{[Rn]}{n\beta}\log \left( \frac{e^{-\beta E}}{1+e^{-\beta E}}\right) - \frac{1}{n\beta} \Tr \left[ \delta_n \log \left( \gamma_P^{\otimes [Rn]} \right)\right] \\
    &\geq -\frac{1}{n\beta} S\left( \ketbra{1}_P^{\otimes [Rn]} + \delta_n\right) + \frac{[Rn]}{n}\, E - \frac{1}{n\beta} \Tr \left[ \delta_n \log \left( \gamma_P^{\otimes [Rn]} \right)\right] ,
\end{aligned}
\label{chain}
\end{equation}
where we remembered that the Gibbs state of the qubit battery is given by Eq.~\eqref{Gibbs battery}. Rearranging and taking the limit $n\to \infty$, we obtain
\begin{equation} \label{upper bound R}
    R E = \lim_{n\to\infty} \frac{[Rn]}{n}\, E \leq W_{a,\infty}^{B|A}(\rho_{AB}) + \lim_{n\to\infty} \frac{1}{n\beta} S\left( \ketbra{1}_P^{\otimes [R n]} + \delta_n\right) + \lim_{n\to\infty} \frac{1}{n\beta} \Tr\left[\delta_n \log \left(\gamma_P^{\otimes [Rn]}\right)\right] ,
\end{equation}
provided that those limits exist. We now claim that the second and third term on the r.h.s.\ of Eq.~\eqref{upper bound R} vanish. Indeed, by Fannes's inequality~\cite{Fannes1973, Audenaert2007, Winter2016} one can write
\begin{equation}
\begin{aligned}
  \frac{1}{n} S\left( \ketbra{1}_P^{\otimes [R n]} + \delta_n\right) &\leq \frac{1}{n} \left( S\left( \ketbra{1}_P^{\otimes [R n]}\right) + \frac12 \epsilon_n \log \left(2^{[Rn]}\right) + h_2\left(\epsilon_n/2\right) \right) \\
  &= \frac{[Rn]}{2n}\,\epsilon_n + \frac1n h_2\left(\epsilon_n/2\right) \\
  &\underset{n\to \infty}{\longrightarrow} 0\, ,
\end{aligned}
\label{limit entropy}
\end{equation}
where in the last step we used the fact that $\lim_{n\to\infty}\epsilon_n = 0$. As for the third term on the r.h.s.\ of Eq.~\eqref{upper bound R}, we observe that
\begin{equation*}
    \log \left(\gamma_P^{\otimes [Rn]}\right) = (\log \gamma_P) \otimes \id\otimes \ldots \otimes \id + \id \otimes (\log \gamma_P) \otimes \ldots \otimes \id +\ldots + \id\otimes \ldots \otimes \id \otimes (\log \gamma_P)\, ,
\end{equation*}
from which one deduces immediately that
\begin{equation*}
    \left\|\log \left( \gamma_P^{\otimes [Rn]}\right)\right\|_\infty = [Rn] \left\|\log \gamma_P\right\|_\infty\, .
\end{equation*}
Applying H\"older's inequality, we then obtain
\begin{equation}
\begin{aligned}
    \frac{1}{n} \Tr\left[\delta_n \log \left(\gamma_P^{\otimes [Rn]}\right)\right] &\leq \frac1n \|\delta_n\|_1 \left\|\log \left( \gamma_P^{\otimes [Rn]}\right)\right\|_\infty \\
    &= \frac{[Rn]}{n}\, \epsilon_n \|\log\gamma_P\|_\infty \\
    &\underset{n\to\infty}{\longrightarrow} 0\, ,
\end{aligned}
\label{limit delta}
\end{equation}
where the last step is made possible by the fact that $\|\log \gamma_P\|_\infty$ is a constant independent of $n$.

We have thus established that the r.h.s.\ of Eq.~\eqref{upper bound R} coincides with the regularized work of assistance $W_{a,\infty}$, which then upper bounds any achievable rate in the expression for the work of collaboration corresponding to the setting where Bob can only access TO. This concludes the proof.

\section{Work of assistance for isotropic states}\label{workass}
\noindent We now continue with an explicit example of how the work of assistance can be computed by considering the following non-trivial family of bipartite states. The isotropic states~\cite{horodecki1999reduction} on a $d\times d$ system appear as
\begin{equation}\label{Isotropic}
\rho_{AB}(\lambda):=\lambda\Phi_d+\frac{(1-\lambda)}{d^2}\id_{AB},
\end{equation}
where $\Phi_d=|\Phi_d\rangle\langle\Phi_d|$ is the maximally entangled state $|\Phi_d\rangle:=\frac{1}{\sqrt{d}}\sum_{i=1}^{d}|ii\rangle$. The isotropic state is a legitimate state ($\rho_{AB}(\lambda)\ge0$) for $-\frac{1}{d^2-1}\le\lambda\le1$ and entangled for $\lambda>\frac{1}{d+1}$.
The choice of an isotropic state allows us to quantify the role of entanglement in the work of assistance through the constant $\lambda$.  \\
\\In this section we will attempt to determine the ensemble $\{p_i,\tilde{\rho}_{B,i}\}$ of $d$-dimensional states obtained when Alice measures her subsystem of an isotropic state $\rho_{AB}(\lambda)$ where $\lambda$ is a priori fixed. We will then attempt to determine the form of this ensemble that maximizes the distillable work on Bob's subsystem. \\
\\Consider a POVM $\{\Pi_{A,i}\}$ such that $p_i=\Tr_{AB}[(\Pi_{A,i}\otimes\id)\rho_{AB}(\lambda)]$ and $p_i\tilde{\rho}_i=\Tr_A[(\Pi_{A,i}\otimes\id)\rho_{AB}(\lambda)]$. The average of this ensemble can be calculated as
\begin{align}
\sum_i p_i\tilde{\rho}_i=&\Tr_A\left[\left(\sum_i\Pi_{A,i}\otimes\id_B\right)\rho_{AB}(\lambda)\right],\nonumber\\
=&\Tr_A\left[\rho_{AB}(\lambda)\right],\nonumber\\
=&\rho_B(\lambda),\nonumber\\
=&\frac{\id_B}{d}\,\,\,\,\forall\,\,\,\,\lambda.
\end{align}
The states in the ensemble appear as,
\begin{align}
\tilde{\rho}_i=&\frac{1}{p_i}\Tr_A\left[\left(\Pi_{A,i}\otimes\id_B\right)\left(\lambda\Phi_d+\frac{(1-\lambda)}{d^2}\id_{AB}\right)\right],\nonumber\\
=&\frac{1}{p_i}\left\{\lambda\Tr_A\left[\left(\sqrt{\Pi_{A,i}}\otimes\id_B\right)|\Phi_d\rangle\langle\Phi_d|\left(\sqrt{\Pi_{A,i}}\otimes\id_B\right)\right]+\frac{1-\lambda}{d^2}\Tr_A\left[\left(\Pi_{A,i}\otimes\id_B\right)\id_{AB}\right]\right\},
\end{align}
Using the identity $M\otimes\id|\Psi\rangle=\id\otimes M^\dagger|\Psi\rangle$, we have
\begin{align}\label{der1}
\tilde{\rho}_i=&\frac{1}{p_i}\left\{\lambda\Tr_A\left[\left(\id_A\otimes\sqrt{(\Pi_{B,i})^T}\right)|\Phi_d\rangle\langle\Phi_d|\left(\id_A\otimes\sqrt{(\Pi_{B_i})^T}\right)\right]+\frac{1-\lambda}{d^2}\Tr\left[\Pi_i\right]\id_B\right\},\nonumber\\
=&\frac{1}{p_i}\left\{\lambda\sqrt{(\Pi_{B,i})^T}\frac{\id_B}{d}\sqrt{(\Pi_{B,i})^T}+\frac{1-\lambda}{d^2}\Tr[\Pi_i]\id_B\right\},\nonumber\\
=&\frac{1}{p_i}\left\{\frac{\lambda}{d}(\Pi_{B,i})^T+\frac{1-\lambda}{d^2}\Tr[\Pi_i]\id_B\right\}.
\end{align}
By imposing $\Tr\tilde{\rho}_i=1$, this yields,
\begin{align}
1=&\Tr\tilde{\rho}_i,\nonumber\\
=&\frac{1}{p_i}\left\{\frac{\lambda}{d}\Tr \Pi_i+\frac{1-\lambda}{d}\Tr \Pi_i\right\},\nonumber\\
=&\frac{1}{p_i}\frac{\Tr \Pi_i}{d}.
\end{align}
Therefore,
\begin{align}\label{prob}
p_i=\frac{\Tr \Pi_i}{d}
\end{align}
and hence,
\begin{align}\label{states1}
\tilde{\rho}_i=\lambda\frac{(\Pi_{B,i})^T}{\Tr \Pi_i}+\frac{1-\lambda}{d}\id_B.
\end{align}
Depending on whether our fixed variable $\lambda$ is positive or negative, as $\Pi_{B,i}\ge0$ we can write the following,
\begin{align}\label{states2}
\tilde{\rho}_i
\Bigg{\{}\begin{array}{l}
\,\,{\geq} \frac{1-\lambda}{d}\,\id_B\,\,\text{for}\,\,\lambda\ge0 \\
\\
\,\,< \frac{1-\lambda}{d}\,\id_B\,\,\text{for}\,\,\lambda<0
\end{array}.
\end{align}
Which means we can write the ensembles states as
\begin{align}
\tilde{\rho}_i=\frac{1-\lambda}{d}\id\pm\delta_i,\,\,\,\delta_i\ge0\,\,\,\forall\,i.
\end{align}
From this ensemble we will attempt to find a POVM on $A$ that generates it. We start from the fact that the average of this ensemble is the maximally mixed state,
\begin{align}
\frac{\id}{d}=&\sum_ip_i\tilde{\rho}_i,\nonumber\\
=&\sum_{i}p_i\left(\frac{1-\lambda}{d}\id\pm\delta_i\right),\nonumber\\
=&\frac{1-\lambda}{d}\id\pm\sum_{i}p_i\delta_i.
\end{align}
From this we can construct our set of POVM operators
\begin{align}
\pm\sum_ip_i\delta_i=&\lambda\frac{\id}{d},\nonumber\\
\implies\,\,\pm\sum_i\frac{d}{\lambda}p_i\delta_i=&\id,
\end{align}
implying that the above operators form a valid POVM if they have the form $\Pi_i:=\pm\frac{d}{\lambda}p_i\delta_i$. \\
\\This POVM can be checked with our current ensemble,
\begin{align}
\left\{p_i=\frac{\Tr \Pi_i}{d},\,\,\,\tilde{\rho}_i=\lambda\frac{(\Pi_i)^T}{\Tr \Pi_i}+\frac{1-\lambda}{d}\id\right\}_i,
\end{align}
in our case,
\begin{align}
p_i=&\frac{\Tr\Pi_{i}}{d}=\pm\frac{1}{d}\Tr\frac{d}{\lambda}p_i\delta_i
=\pm p_i\frac{\Tr\delta_i}{\lambda},\nonumber\\
=&p_i\frac{\Tr\left[\tilde{\rho}_i-\frac{1-\lambda}{d}\id\right]}{\lambda}=p_i\frac{1-(1-\lambda)}{\lambda},\nonumber\\=&p_i,
\end{align}
and
\begin{align}
\tilde{\rho}_i=&\lambda\frac{(\Pi_i)^T}{\Tr \Pi_i}+\frac{1-\lambda}{d}\id=\lambda\frac{\pm\frac{d}{\lambda}p_i\delta_i}{\pm\frac{d}{\lambda}p_i\Tr\delta_i}+\frac{1-\lambda}{d}\id,\nonumber\\
=&\lambda\frac{\pm\delta_i}{\pm\Tr\delta_i}+\frac{1-\lambda}{d}\id=\lambda\frac{\tilde{\rho}_i-\frac{1-\lambda}{d}\id}{\lambda}+\frac{1-\lambda}{d}\id,\nonumber\\
=&\tilde{\rho}_i-\frac{1-\lambda}{d}\id+\frac{1-\lambda}{d}\id,\nonumber\\
=&\tilde{\rho}_i.
\end{align}
confirming our choice of POVM. \\
\\In order to use this conceived ensemble in the work of assistance~\eqref{workass1} it must minimize the entropy. In order for a  mixed state to minimize its entropy it must concentrate its spectrum on a single eigenvalue, making it as pure as possible. The conditions on these states~\eqref{states1} and probabilities~\eqref{prob} are that their average is the maximally mixed state and that equation~\eqref{states2} is satisfied for positive or negative fixed values of $\lambda$.\\
\\An ensemble for which all the above holds is the following,
\begin{align}
\left\{p_i=\frac{1}{d},\,\,\,\tilde{\rho}_i=\frac{1-\lambda}{d}\id+\lambda|i\rangle\langle i|\,\,\,\forall\,i\right\}_i.
\end{align}
which, using $\lambda+\frac{1-\lambda}{d}=1-(d-1)\frac{1-\lambda}{d}$, allows us to write down the work of assistance for this family of states,
\begin{align}
W_a^{B|A}(\rho_{AB}(\lambda))=-\frac{1}{\beta}\left[\Tr\frac{1}{d}\log \gamma_B+H\left(\overbrace{\frac{1-\lambda}{d},...,\frac{1-\lambda}{d}}^{d-1},1-(d-1)\frac{1-\lambda}{d}\right)\right],
\end{align}
where $H(p_1,...,p_d)=-\sum_ip_i\log p_i$ is the Shannon entropy.

It is known~\cite{King2003}~\cite[p.~448]{HAYASHI} that the Henderson-Vedral measure $J^{\rightarrow}$ and hence the work of assistance~\eqref{workass4} is additive over the family of isotropic states.

\section{Monotonicity and rewriting of $W_{r}^{B|A}(\rho_{AB})$}\label{rewrite}
We start with the following definition,
\begin{align}
W_{r}^{B|A}(\rho_{AB})=\frac{1}{\beta}\min_{\Gamma_{AB}\in {QT}}S(\rho_{AB}||\Gamma_{AB}),
\end{align}
where the minimization is taken over the set ${QT}$. \\
As the thermal state $\gamma_B$ appearing within the QT state $\Gamma_{AB}=\sigma_{A}\otimes\gamma_B$ is fixed, the minimization should be taken over Alice's state $\sigma_A$,
\begin{align}\label{mutual2}
W_{r}^{B|A}(\rho_{AB})=&\frac{1}{\beta}\min_{\sigma_A}S(\rho_{AB}||\sigma_A \otimes \gamma_B)\nonumber\\
=&\frac{1}{\beta}\min_{\sigma_A}\left\{-\Tr\rho_{AB}(\log\sigma_A+\log\gamma_B)-S(\rho_{AB})\right\}\nonumber\\
=&\frac{1}{\beta}\left(-\left(\max_{\sigma_A}\Tr\rho_A \log\sigma_A \right)-\Tr\rho_B \log\gamma_B-S(\rho_{AB})\right).
\end{align}
Now using the following, $S(\rho_A)=-\max_{\sigma_A}\Tr\rho_A \log\sigma_A $, as $\min_{\sigma_A}S(\rho_A ||\sigma_A)=0$, in~\eqref{mutual2}, we get
\begin{align}\label{mutual}
W_{r}^{B|A}(\rho_{AB})=&\frac{1}{\beta}\left(-S(\rho_{AB})+S(\rho_A)-\Tr\rho_B \log\gamma_B\right)\nonumber\\
=& \frac{1}{\beta}S(\rho_{AB}||\rho_A \otimes \gamma_B).
\end{align}
If we take the bipartite quantum state $\Lambda(\rho_{AB})=\tilde{\rho}_{A'B}$ where $\Lambda$ are the set of allowed operations, we can write,
\begin{align}
	W_{r}^{B|A}\left(\tilde{\rho}_{A'B}\right)=&\frac{1}{\beta}\inf_{\tilde{\sigma}_{A'}}S\left(\tilde{\rho}_{A'B}||\tilde{\sigma}_{A'} \otimes \gamma_B\right)\nonumber\\
	\le&\frac{1}{\beta}S\left(\tilde{\rho}_{A'B}||\rho_A\otimes\gamma_B\right)\nonumber\\
	\le&\frac{1}{\beta}S\left(\rho_{AB}||\rho_A\otimes\gamma_B\right)\nonumber\\
	=&	W_{r}^{B|A}\left(\rho_{AB}\right),
\end{align}
where $\Lambda\left(\rho_A\otimes\gamma_B\right)=\tilde{\sigma}_{A'} \otimes \gamma_B$, demonstrating the monotonicity of $W_{r}^{B|A}(\rho_{AB})$ under the set of allowed operations. Also due to the additivity property of relative entropy, $W_{r}^{B|A}(\rho_{AB})$ is also additive.

We also note that $W_{r}^{B|A}(\rho_{AB})$ can be expressed in relation to the quantum mutual information,
\begin{align}\label{rel}
W_{r}^{B|A}(\rho_{AB})=&\frac{1}{\beta}\left(-\Tr \rho_{AB}\log(\rho_A\otimes\gamma_B)-S(\rho_{AB})\right)\nonumber\\
=&\frac{1}{\beta}\left(-\Tr \rho^{A}\log(\rho_A)-\Tr \rho_B\log(\gamma_B)-S(\rho_{AB})\right)\nonumber\\
=&\frac{1}{\beta}\left(S(\rho^{A}||\rho_A)+S(\rho_B||\gamma_B)+I(\rho_{AB})\right),\nonumber\\
=&\frac{1}{\beta}\left(S(\rho_B||\gamma_B)+I(\rho_{AB})\right).
\end{align}
which is a measure of the shared total correlations between the two parties.

\section{Upper bound on the work of collaboration}\label{UpperColab}
In this section we will prove that the work of collaboration is upper bounded by $W_{r}^{B|A}(\rho_{AB})$,
\begin{align}
W_{r}^{B|A}(\rho_{AB})\ge W_{c}^{B|A}(\rho_{AB}).
\end{align}
This proof will follow analogously to~\cite{vedral1998entanglement,chitambar2016assisted}.
Let $R$ be an achievable rate for Eq.~\eqref{collab1}, so that there exists a sequence of protocols $\Lambda_n\in \mathcal{O}_c^{B|A}$ with the property that $\Lambda_n \left( \rho_{AB}^{\otimes n} \otimes \ketbra{0}_P^{\otimes [Rn]}\right) = \ketbra{1}_P^{\otimes [Rn]} + \delta_n$, where $\lim_{n\to\infty} \|\delta_n\|_1=0$. Using the monotonicity of $W_r^{B|A}$ together with its additivity, and repeating the steps in Eq.~\eqref{chain}, we obtain that
\begin{align*}
    W_r^{B|A}(\rho_{AB}) &= \frac1n W_r^{B|A}\left( \rho_{AB}^{\otimes n}\otimes \ketbra{0}_P^{\otimes [Rn]}\right) \\
    &\geq \frac1n W_r^{B|A}\left(\ketbra{1}_P^{\otimes [Rn]} + \delta_n\right) \\
    &= \frac{1}{\beta n} S\left( \ketbra{1}_P^{\otimes [Rn]} + \delta_n \Big\|\gamma_P^{\otimes [Rn]}\right) \\
    &\geq -\frac{1}{\beta n} S\left( \ketbra{1}_P^{\otimes [Rn]} + \delta_n\right) + \frac{[Rn]}{n} E - \frac{1}{n\beta} \Tr \left[ \delta_n \log \left( \gamma_P^{\otimes [Rn]} \right)\right] .
\end{align*}
Note that the equality in the third line holds because the battery $P$ is on Bob's side. Employing the already established Eq.~\eqref{limit entropy} and~\eqref{limit delta}, we arrive at the bound $RE\leq W_r^{B|A}(\rho_{AB})$, concluding the proof.

\end{document}